\documentclass[11pt]{article}
\usepackage[letterpaper,margin=1in]{geometry}
\usepackage{amsmath,amsthm,amssymb}
\usepackage{graphicx}
\usepackage[ruled,vlined]{algorithm2e}
\usepackage{paralist}
\usepackage{enumitem}
\usepackage{tikz}
\usepackage{verbatim}
\usetikzlibrary{shapes,arrows,fit,calc,positioning}
\usetikzlibrary{arrows}
\usetikzlibrary{decorations.markings}

\SetAlFnt{\small}
\SetAlCapFnt{\small}
\SetAlCapNameFnt{\small}
\SetAlCapHSkip{0pt}
\IncMargin{-\parindent}

\def \S {\boldsymbol\Sigma}
\def \P {\mathbf{P}}
\def \cP {\mathcal{P}}
\def \O {\mathbf{O}}

\def \< {\langle}
\def \> {\rangle}
\def \x {\mathbf{x}}
\def \y {\mathbf{y}}
\def \z {\mathbf{z}}

\DeclareMathOperator*{\E}{\mathbb{E}}

\theoremstyle{definition}
\newtheorem{theorem}{Theorem}

\newtheorem{lemma}[theorem]{Lemma}

\newtheorem{definition}[theorem]{Definition}

\newtheorem{corollary}[theorem]{Corollary}
\newtheorem{example}[theorem]{Example}
\theoremstyle{remark}
\newtheorem{claim}[theorem]{Claim}

\newcounter{note}[section]

\title{Designing Networks with Good Equilibria under Uncertainty}

\author{George Christodoulou\thanks{University of Liverpool,
    UK. Email:\texttt{\{gchristo,salkmini\} @liverpool.ac.uk} }
		\thanks{This author was supported by EPSRC grants EP/M008118/1 and EP/K01000X/1.}
  \and Alkmini Sgouritsa\footnotemark[1] }

\date{}

\begin{document}
\thispagestyle{empty}
\maketitle

\begin{abstract}

  We consider the problem of designing network cost-sharing protocols
  with good equilibria under uncertainty. The underlying game is a
  {\em multicast game} in a rooted undirected graph with nonnegative
  edge costs. A set of $k$ {\em terminal vertices} or {\em players}
  need to establish connectivity with the root. The social optimum is
  the Minimum Steiner Tree.

  We are interested in situations where the designer has incomplete
  information about the input. We propose two different models, the
  {\em adversarial} and the {\em stochastic}. In both models, the
  designer has prior knowledge of the underlying metric but the
  requested subset of the players is {\em not known} and is
  activated either in an adversarial manner (adversarial model) or is
  drawn from a known probability distribution (stochastic model).

  In the adversarial model, the goal of the designer is to choose a single,
  {\em universal} cost-sharing protocol that has low Price of Anarchy
  (PoA) for {\em all possible} requested subsets of players. The main
  question we address is: {\em to what extent can prior knowledge of
    the underlying metric help in the design?}

  We first demonstrate that there exist classes of graphs where
  knowledge of the underlying metric can dramatically improve the
  performance of good network cost-sharing design. For {\em
    outerplanar} graph metrics, we provide a universal cost-sharing
  protocol with constant PoA, in contrast to protocols that, by
  ignoring the graph metric, cannot achieve PoA better than
  $\Omega(\log k)$.  Then, in our main technical result, we show that
  there exist graph metrics, for which knowing the underlying metric
  does not help and {\em any universal} protocol has PoA of
  $\Omega(\log k)$, which is tight. We attack this problem by
  developing new techniques that employ powerful tools from extremal
  combinatorics, and more specifically Ramsey Theory in high
  dimensional hypercubes.
	
  Then we switch to the stochastic model, where each player is
  independently activated according to some probability distribution
  that is known to the designer. We show that there exists a {\em
    randomized} ordered protocol that achieves constant PoA. By using
  standard derandomization techniques, we produce a {\em
    deterministic} ordered protocol that achieves constant PoA. We
  remark, that the first result holds also for the {\em black-box}
  model, where the probabilities are not known to the designer, but is
  allowed to draw independent (polynomially many) samples.
	
\end{abstract}


\thispagestyle{empty}\setcounter{page}{0}
\clearpage \newpage

\section{Introduction}

\paragraph{Network Cost-Sharing Games.}
We study a {\em multicast game} in a rooted undirected graph $G=(V,E)$
with a nonnegative cost $c_e$ on each edge $e\in E$. A set of $k$ {\em
  terminal vertices} or {\em players} $s_1,\ldots, s_k$ need to
establish connectivity with the root $t$. Each player selects a path
$P_i$ and the outcome produced is the graph $H=\cup_iP_i$. The global
objective is to minimize the cost $\sum_{e\in H}c_e$ of this graph,
which is the {\em Minimum Steiner Tree}.

The cost of an edge may represent infrastructure cost for establishing
connectivity or renting expense, and needs to be covered by the
players that use that edge in the solution. There are several ways to
split the edge costs among the users and this is dictated by a {\em
  cost-sharing protocol}. Naturally, it is in the players best
interest to choose paths that charge them with small cost, and
therefore the solution will be a Nash equilibrium (NE).
Algorithmic Game
Theory provides tools to analyze the quality of the equilibrium
solutions; this can be measured with the Price of Anarchy
(PoA)~\cite{KP99} (or Price of Stability (PoS)~\cite{ADKTWR08}) that
compares the worst-case (or the best-case) cost in a Nash equilibrium
with the cost of the minimum Steiner tree.
This is a fundamental network design game that was originated by
Anshelevich et al.~\cite{ADKTWR08} and has been extensively studied
since. \cite{ADKTWR08} studied the {\em Shapley} cost-sharing protocol,
where the cost of each edge is equally split among its users. They
showed that the quality of equilibria can be really poor\footnote{Even for simple
networks the PoA grows linearly with the number of players, $k$. The
PoS is not well-understood. It is a big open question to determine its
exact value that is between constant and
$O(\log/\log\log k)$~\cite{Li09}.}.

\paragraph{Cost-Sharing Protocol Design.}
Different cost-sharing protocols result in different quality of
equilibria. In this work, we are interested in the design of protocols
that induce good equilibrium solutions in the {\em worst-case},
therefore we focus on protocols that guarantee low PoA. 
Chen, Roughgarden and Valiant~\cite{CRV10} were the first to
address design questions for network cost-sharing games. They gave a
characterization of protocols that satisfy some natural axioms and
they thoroughly studied their PoA for the following two classes of
protocols, that use different informational assumptions from the
perspective of the designer.
\vspace{-1mm}
\begin{description}
\item{\em Non-uniform protocols}. The designer has full knowledge of
  the instance, that is, she knows both the network topology given by
  $G$ and the costs $c_e$, and in addition the set of players'
  requests $s_1, \ldots, s_k$. They showed that a simple priority protocol has a
  constant PoA; the Nash equilibria induced by the protocol simulate
  Prim's algorithm for the Minimum Spanning Tree (MST) problem, and
  therefore achieve constant approximation.
\item{\em Uniform protocols}. The designer needs to decide how to
  split the edge cost among the users {\em without knowledge of the
    underlying graph}. They showed that the PoA is $\Theta(\log k)$;
  both upper and lower bound comes from the analysis of the Greedy
  Algorithm for the Online Steiner Tree problem (OSTP).
\end{description}

\paragraph{Cost-Sharing Design  under Uncertainty.}
Arguably, there are situations where the former assumption is too
optimistic while the latter is too pessimistic. 
We propose a model that lies in the middle-ground, as a framework to
design network cost-sharing protocols with good equilibria, when the
designer has {\em incomplete information}.

We assume that the designer has prior knowledge of the underlying
metric, (given by the graph $G$ and the shortest path metric induced
by the costs $c_e$), but is {\em uncertain} about the requested subset
of players. We consider two different models, the {\em adversarial
  model} and the {\em stochastic model}. In the former, the designer
{\em knows nothing} about the number or the positions of the $s_i$'s
and has as goal to process the graph and choose a single, {\em
  universal} cost-sharing protocol that has low PoA against {\em all
  possible} requested subsets. Here, no distributional assumptions are
made about arrivals of players, and take the worst-case approach similarly to
Competitive Analysis. Once the designer selects the protocol, then an
adversary will choose the requested subset of players and their
positions in the graph (the $s_i$'s), in a way that {\em maximizes}
the PoA of the induced game. In the stochastic model, the
players/vertices are activated according to some probability
distribution which is given to the designer. The goal is now to choose
a universal protocol where the expected worst-case cost in the Nash
equilibrium is not far from the expected optimal cost.

\begin{example}{\bf (Ordered protocols)}. An important special class
  with interesting properties is that of {\em ordered protocols}. The
  designer decides a total order of the users, and when a subset of
  players uses some edge, the full cost is covered by the player who
  comes first in the order. Any NE of the induced game corresponds to
  the solution produced by the Greedy Algorithm for the MST: each
  player is connected, via a shortest path, with the component of the
  players that come before him in the order. The analysis of the PoA
  in the uniform model boils down to the analysis of the Greedy
  Algorithm for the OSTP, where the worst-case order is considered.
  The following example demonstrates that even this special class of
  ordered protocols becomes very rich, once the designer has prior
  knowledge of the underlying metric space.  Uniform protocols throw
  away this crucial component, the structure of the underlying metric,
  that universal protocols can use in their favor to come up with
  better PoA guarantees.

{\footnotesize
\begin{figure}[t]
\centering
\begin{tabular}{c c c}
\includegraphics[scale=0.3]{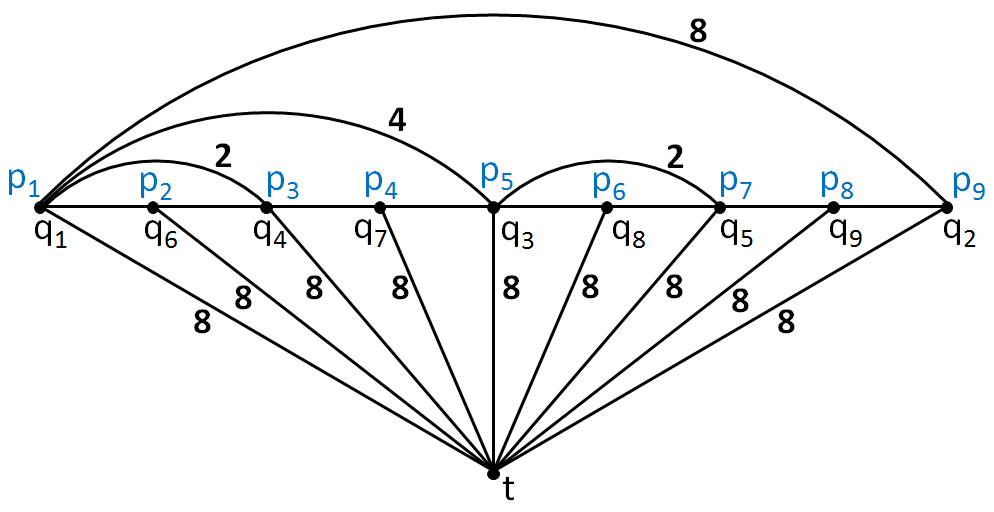} &
\includegraphics[scale=0.38]{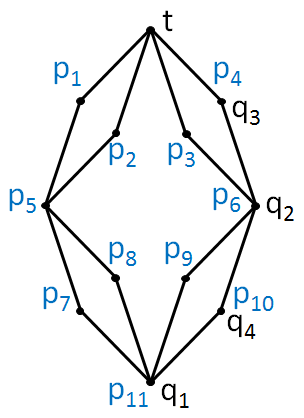} &
\includegraphics[scale=0.35]{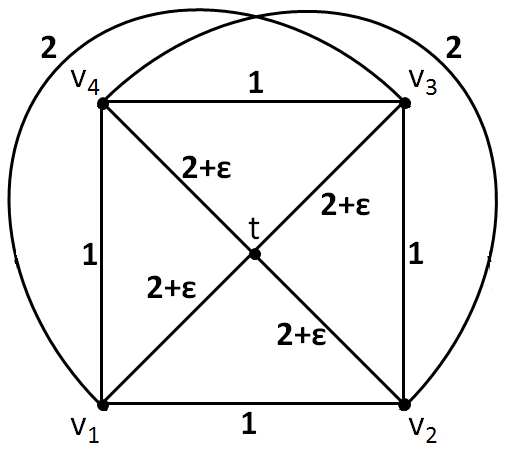}\\
$(a)$&$(b)$&$(c)$
\end{tabular}
 \caption{\footnotesize In $(a)$ and $(b)$ we assume two orders on the vertices,
  denoted by $q_i$ or $p_i$. The $q$-order is adversarially chosen and
  simulates the adversary for the OSTP~\cite{IW91}, that
  results to high PoA of $\Omega(\log k)$. The $p$-order results to
  constant PoA. $(c)$ shows an example where both the {\em best}
  ordered protocol and the Shapley protocol have PoA $\geq 5/4$,
  whereas there is an {\em intermediate} protocol with PoA $1$. In edges 
	with no written cost, we consider the unit cost; we take $\varepsilon > 0$ 
	arbitrarily small.}
     \label{fig:line_diam_sq}
\end{figure}
}
\vspace{-1mm}
\begin{description}
\item{Uniform protocols}. The designer chooses an order of the players
  $1,\ldots, k$ without prior knowledge of the graph. The adversary
  constructs a worst-case graph, by simulating the adversary for the
  Greedy Algorithm of the OSTP~\cite{IW91}, and places the players accordingly
  (see for example Figure~\ref{fig:line_diam_sq}(a),(b), the $q$
  labels). Therefore the PoA of uniform ordered protocol is $\Omega
  (\log k)$~\cite{CRV10}.
\item{Universal protocols}. The designer takes into account the graph;
  consider the worst-case graph for the Greedy Algorithm of the OSTP
  (illustrated in Figure~\ref{fig:line_diam_sq}(a),(b) for a small
  number of players). For the graph of
  Figure~\ref{fig:line_diam_sq}(a), choose the linear order dictated
  from the path $p_1,\ldots, p_9$ (say from left to right). For the
  graph of Figure~\ref{fig:line_diam_sq}(b) order the vertices
  according to their distance from $t$, $p_1,\ldots, p_{11}$. The
  adversary will choose $k$ and the positions of the players
  ($s_1,\ldots, s_k$). In both cases, it is not hard to see that, {\em
    no matter which subset of players the adversary chooses}, the PoA
  remains constant as $k$ grows.
\end{description}
\vspace{-2mm}

\label{ex:order}
\end{example}

\begin{example}{\bf (Generalized weighted Shapley).}
  In \cite{CRV10}, it was shown that ordered protocols are essentially
  optimal among uniform protocols. In our model, the choice of the
  optimal method may depend on the underlying graph metric. Take the
  example in Figure~\ref{fig:line_diam_sq}(c). By using Shapley cost sharing 
	the adversary can choose $v_1,v_2,v_3$ and in the Nash equilibrium $v_1$, 
	$v_3$ connect directly to $t$ and $v_2$ connects through $v_1$. Regarding 
	{\em any} ordered protocol, the square defined by the $v_i$'s
  contains a path of length $2$ where the middle vertex comes {\em
    last} in the order. The adversary will select this triplet of
  players, say $v_1,v_2,v_3$. In the Nash equilibrium, $v_1$ connects 
	directly to $t$, $v_3$ and $v_2$ connect through $v_1$.  
	In both cases, the cost of the Nash equilibria is $5$ and the minimum Steiner tree 
	that connects those vertices with $t$ has cost $4$ (by ignoring $\varepsilon$) 
	and therefore, PoA $\geq 5/4$. 

        However the following (generalized Shapley) protocol, has
        $PoA=1$. Partition the players into two sets
        $S_1=\{v_1,v_2\},S_2=\{v_3,v_4\}$. If players from both
        partitions appear on some edge, then the cost is charged only
        to players from $S_1$.  Players that belong to the same
        partition share the cost equally. One can verify that for all
        possible subsets of players this protocol produces only
        optimal equilibria.
\label{ex:GWS}
\end{example}

\paragraph{Results.}
We propose a framework for the design of (universal) network
cost-sharing protocols with good equilibria, in situations where the
designer has incomplete information about the input.  We consider two
different models, the {\em adversarial} and the {\em stochastic}. In
both models, the designer has prior knowledge of the underlying metric
but the requested subset of the players is {\em not known} and is
activated either in an adversarial manner (adversarial model) or is
drawn from a known probability distribution (stochastic model). The
central question we address is: {\em to what extent does prior
  knowledge of the metric help in good network design under
  uncertainty?}

For the adversarial model, we first demonstrate that there exist
classes of graph metrics where prior knowledge of the underlying
metric can dramatically improve the performance of good network
cost-sharing design. For {\em outerplanar} graph metrics, we provide a
universal ordered cost-sharing protocol with constant PoA, against any
choice of the adversary.  This is in contrast to uniform protocols
that ignore the graph and cannot achieve PoA better than $\Omega (\log
k)$ in outerplanar metrics.

Our main technical result shows that there exist graph metrics, for
which knowing the underlying metric does not help the designer, and
{\em any universal} protocol has PoA of $\Omega(\log k)$. This matches
the upper bound of $O(\log k)$ that can be achieved without prior
knowledge of the metric~\cite{IW91,CRV10}. 

Then we switch to the stochastic model, where each player is
independently activated according to some probability that is known to
the designer. We show that there exists a {\em randomized} ordered
protocol that achieves constant PoA. By using standard
derandomization techniques~\cite{WZ07,ST08}, we produce a {\em
  deterministic} ordered protocol that achieves constant PoA. We
remark, that the first result holds also for the {\em black-box}
model, where the probabilities are not known to the designer, but is
allowed to draw independent (polynomially many) samples.

Our results for the adversarial model motivate the following question
that is left open.
 
\noindent {\bf Open Question:} For which metric spaces can one design
universal protocols with constant PoA? What sort of structural graph
properties are needed to obtain good guarantees?

\vspace{5pt}
\paragraph{Techniques.} 
We prove our main lower bound for the adversarial model in two parts. In the first part
(Section~\ref{generalLB}) we bound the PoA achieved by
{\em any ordered} protocol. 
Our origin is a well-known ``zig-zag'' {\em ordered} structure that
has been used to show a lower bound on the Greedy Algorithm of the
OSTP (see the labeled path $(q_1,q_6,q_4,\ldots,q_2)$ in
Figure~\ref{fig:line_diam_sq}(a)). The challenge is to show that high
dimensional hypercubes exhibit such a distance preserving structure
{\em no matter how the vertices are ordered}. Section~\ref{generalLB} is devoted to this and we believe that this is
of independent interest.
	 
We show the existence proof by employing powerful tools from Extremal
Combinatorics and in particular Ramsey Theory~\cite{Ram90}. We are
inspired by a Ramsey-type result due to Alon et al.~\cite{AlonRSV06},
in which they show that for any given length $\ell\geq 5$, any
$r$-edge coloring of a high dimensional hypercube contains a
monochromatic cycle of length $2\ell$. Unfortunately, we cannot
immediately use their results, but we show a similar Ramsey-type
result for a different, carefully constructed structure; we assert
that every 2-edge coloring of high dimensional hypercubes $Q_n$
contains a monochromatic copy of that structure. Then, we prescribe a
special $2$-edge-coloring that depends on the ordering of $Q_n$, so
that the special subgraph preserves some nice labeling
properties. 
A suitable subset of the subgraph's vertices can be 1-embedded into a
hypercube of lower dimension. Recursively, we show existence of the
desired {\em distance preserving} structure.
 
In the second part (Section~4), we extend the lower bound to {\em all
  universal} cost-sharing protocols, by using the characterization of
\cite{CRV10}. At a high level, we use as basis the construction for
the ordered protocol and create ``multiple copies"\footnote{Note that
  the standard complexity measure, to analyze the inefficiency of
  equilibria, is the number of participants, $k$, and not the total
  number of vertices in the graph (see for example
  \cite{ADKTWR08,CRV10}).}. The adversary will choose different
subsets of players, depending on whether the designer chose protocols
``closer'' to Shapley or to ordered. In the latter case, we use
arguments from Matching Theory to guarantee existence of ordered-like
players in one of the hypercubes.

For the stochastic model (Section~\ref{sec:stochastic}), we construct
an approximate minimum Steiner tree over a subset of vertices which
are drawn from the known probability distribution. This tree is used
as a base to construct a spanning tree, which determines a total order
over the vertices. We finally produce a deterministic order by
applying standard derandomization techniques~\cite{WZ07,ST08}.


\paragraph{Related Work}
Following the work of \cite{ADKTWR08,ADTW08}, a long line of research
studies network cost-sharing games, mainly focusing on the PoS of the
Shapley cost-sharing mechanism. \cite{ADKTWR08} showed a tight
$\Theta(\log k)$ bound for directed networks, while for undirected
networks several variants have been studied \cite{BB11,BFM13,CR09,CCLPS09,DFKM13,FKLOS06,Li09,BCFM13} but the exact value of PoS
still remains a big open problem. For multicast games, an improved
upper bound of $O(\log k/\log\log k)$ is known due to Li~\cite{Li09},
while for broadcast games, a series of work \cite{FKLOS06,LL13} lead
finally to a constant due to Bil\`{o} et al.~\cite{BFM13}. The PoA of
some special equilibria has been also studied
in~\cite{CKMNS08,CCLNO07}.

Chen, Roughgarden and Valiant \cite{CRV10} initiated the study of
network cost-sharing design with respect to PoA and PoS. They
characterized a class of protocols that satisfy certain desired
properties (which was later extended by Gopalakrishnan, Marden and
Wierman, in \cite{GMW13}), and they thoroughly studied PoA and PoS for
several cases.  
Falkenhausen and Harks~\cite{FH13} studied parallel links and weighted
players while Gkatzelis, Kollias and Roughgarden~\cite{GKR14}, focus on weighted
congestion games with polynomial cost functions. 

Close in spirit to universal cost-sharing protocols is the notion of
Coordination Mechanisms~\cite{CKN09} that provides a way to improve the
PoA in cases of incomplete information. The designer has to decide in
advance local scheduling policies or increases in edge latencies,
without knowing the exact input, and has been used for scheduling
problems~\cite{CKN09,ILMS05,Kol08,AJM07,Caragiannis09,CCG+11,BIKM14,AbedCH14,AbedH12} as well as for simple routing
games~\cite{CMP14,BKM14}.
 
As discussed in
Example~\ref{ex:order}, the analysis of the equilibria induced by
ordered protocols corresponds to the analysis of Greedy Algorithm for
the MST. In the uniform model, this corresponds to the analysis of the
Greedy Algorithm~\cite{IW91,AAB04} for the (Generalized) OSTP \cite{AA93,BC97,Umb15},
which was shown to be $\Theta(\log k)$-competitive by
Imase and Waxman~\cite{IW91} ($O(\log^2 k)$-competitive for the Generalized OSTP by \cite{AAB04}). The universal
model is closely related to universal network design
problems~\cite{JLNRS05}, hence our choice for the term
``universal''. In the universal TSP, given a metric space, the
algorithm designer has to decide a {\em master} order so that tours
that use this order have good approximation~\cite{PB89,BG88,HKL06,GKSS10,JLNRS05}.

Much work has been done in stochastic models and we only mention the
most related to our work. Karger and Minkoff \cite{KM00} showed a
constant approximation guarantee for the maybecast problem, where the
designer needs to fix (before activation) some path for every vertex
to the root. Garg et al. \cite{GGLS08} gave bounds on the
approximation of the stochastic online Steiner tree problem.  A line
of works \cite{BJO90,GKSS10,SS08,ST08} has studied the {\em a priori
  TSP}. Shmoys and Talwar \cite{ST08} assumed independent activations
and demonstrated randomized and deterministic algorithms with constant
approximations.


\section{Model and definitions}
\label{sec:prel}

\paragraph{Universal Cost-Sharing Protocols.}
A {\em multicast network cost-sharing game}, is specified by a
connected undirected graph $G=(V,E)$, with a designated root $t$ and
nonnegative weight $c_e$ for every edge $e$, a set of players
$S=\{1,\ldots, k\}$ and a cost-sharing protocol. Each player $i$ is
associated with a {\em terminal\footnote{We abuse notation and use $S$
    to refer both to the players and their associated vertices.}}
$s_i$, which she needs to connect with $t$. We say that a vertex is
{\em activated} if there exists some requested player associated with
it.  In the {\em adversarial model} the designer {\em knows nothing}
about the set $S$ of activated vertices, while in the {\em stochastic
  model}, the vertices are activated according to some probability
distribution $\Pi$ which is known to the designer.

For any set $N$ of players, a {\em cost-sharing method}
$\xi:2^N\rightarrow \mathbb{R}_+^{|N|}$ decides, for every subset $R
\subseteq N$, the cost-share $\xi(i,R)$ for each player $i\in R$. A
natural rule is that the shares for players not included in $R$ should
always be $0$, i.e.  if $i \notin R$, $\xi(i,R) = 0$.  W.l.o.g. each
player is associated with a distinct vertex\footnote{ To see this, if
  there are two players with $s_1= s_2 = v$, for some $v\in V$, we
  modify the graph by connecting a new vertex $v'$ with $v$ via a
  zero-cost edge and then we set $s_1= v$ and $s_2 = v'$. Neither the
  optimum solution, nor any Nash equilibrium are affected by this
  modification.}.  For any graph $G$ and any set of players $N$, a
{\em cost-sharing protocol} $\Xi$ assigns, for every $e \in E$, some
cost-sharing method $\xi_e$ on $N$.

Following previous work~\cite{CRV10,FH13}, we focus on cost-sharing
protocols that satisfy the following natural properties:
\begin{enumerate}[label={(\arabic*)}]
\item {\em Budget-balance:} For every network
game induced by the cost sharing protocol $\Xi$, and 
every outcome of it, 
$\sum_{i \in R}\xi_e(i,R) = c_e$, for every edge $e$ with cost $c_e$.

\item {\em Separability:} For every network
game induced by the cost sharing protocol $\Xi$, the cost shares of each edge
are completely determined by the set of players using it.

\item {\em Stability:} In every network game induced by the cost-sharing 
protocol $\Xi$, there exists at least one pure Nash equilibrium, regardless of the graph structure.
\end{enumerate} 

We call a cost-sharing protocol $\Xi$ {\em universal}, if it satisfies
the above properties for any graph $G$, and it assigns the
cost-sharing method $\xi_e:2^V \rightarrow \mathbb{R}_+^{|V|}$ to edge $e$
based only on knowledge of $G$ (without knowledge of $S$\footnote{The methods should be defined on $V$, since every vertex
  is potentially associated with some player.}) for the adversarial
model, while in the stochastic the method can in addition depend on
$\Pi$. Due to the characterization in \cite{CRV10}, we restrict
ourselves to the family of generalized weighted Shapley
protocols\footnote{\cite{CRV10} characterizes the linear protocols
  (for every edge $e$ of cost $c_e \geq 0$, it assigns the method $c_e
  \cdot \xi$, where $\xi$ is the method it assigns to any edge of unit
  cost) to be the generalized weighted Shapley protocols. They further
  showed that for any non-linear protocol, there exists a linear one
  with at most the same PoA.}.

\paragraph{Generalized Weighted Shapley Protocol (GWSP).} 
The {\em generalized weighted Shapley protocol} ({\em GWSP}) is defined by 
the players' weights (parameters) $\{w_1,\ldots,w_n\}$ 
and an {\em ordered} partition of the players $\S = (U_1, \ldots, U_h)$. 
An interpretation of $\S$ is that for $i<j$, players from $U_i$ ``arrives" before players from $U_j$.
More formally, for every edge $e$ of cost $c_e$, every set of players $R_e$ that uses $e$ and 
for $s=\arg \min_j\{U_j|U_j \cap R_e \neq \emptyset\}$, the GWSP assigns the following method to $e$:
\[\xi_e(i,R_e)=\left\{ 
	\begin{array}{ll}
		\frac{w_i}{\sum_{j \in U_s \cap R_e} w_j} c_e, \;\; & \mbox{if } i \in U_s \cap R_e \\
		0,\;\; & \mbox{otherwise }
	\end{array}
\right.\]

In the special case that each $U_i$ 
contains exactly one player, the protocol is called {\em ordered}. The order of the $U_i$ sets 
indicates a permutation of the players, denoted by $\pi$. 

\paragraph{(Pure) Nash Equilibrium (NE).}  We denote by $\cP_i$ the
strategy space of player $i$, i.e. the set of all the paths connecting
$s_i$ to $t$.  $\P = (P_1, \ldots, P_k)$ denotes an {\em outcome} or a
{\em strategy profile}, where $P_i \in \cP_i$ for all $i \in S$. As
usual, $\P_{-i}$ denotes the strategies of all players but
$i$. 
Let $R_e$ be the set of players using edge $e \in E$ under $\P$. The
cost share of player $i$ induced by $\xi_e$'s is equal to $c_i(\P) =
\sum_{e \in P_i} \xi_e(i,R_e)$.  The players' objective is to minimize
their cost share $c_i(\P)$.  A strategy profile $\P = (P_1, \ldots,
P_k)$ is a {\em Nash equilibrium (NE)} if for every player $i \in S$
and every strategy $P'_i \in \cP_i$, $c_i(\P) \leq c_i(\P_{-i},P'_i).$
 
\paragraph{Price of Anarchy (PoA).} 
The cost of an outcome
$\P=(P_1,\ldots,P_k)$ is defined as $c(\P) = \sum_{e \in \cup_i P_i}
c_e$, while $\O=(O_1,\ldots,O_k) \in \arg \min_{\P} c(\P)$ is the
optimum solution.
The {\em Price of Anarchy (PoA)} is defined as the worst-case ratio of
the cost in a NE over the optimal cost in the game induced by $S$. In
the adversarial model the {\em worst-case} $S$ is chosen, while
in the stochastic model $S$ is drawn from a known distribution
$\Pi$. Formally, in the adversarial model we define the PoA of a
protocol $\Xi$ on $G$ as

$$PoA(G,\Xi) = \max_{\substack{S\subseteq V\setminus\{t\}}} \frac{\max_{\P \in
  \;\mathcal{N}} c(\P)}{c(\O)}, $$
where $\mathcal{N}$ is the
set of all NE of the game induced by $\Xi$ and $S$ on $G$.

In the stochastic model, the PoA of $\Xi$, given $G$ and
$\Pi$ is
$$ PoA(G,\Xi, \Pi) =  \frac{\E_{S\sim \Pi} \left[\max_{\P \in \;\mathcal{N}}
    c(\P)\right]}{\E_{S\sim \Pi}[c(\O)]}. $$

In both models the objective of the designer is to come up with
protocols that {\em minimize} the above ratios.
Finally, the Price of Anarchy for a class of graph metrics $\mathcal{G}$,
is defined as 
$$PoA(\mathcal{G}) = \max_{G\in \mathcal{G}} \min_{\Xi(G)} PoA(G,\Xi); \qquad 
PoA(\mathcal{G}) = \max_{\substack{G \in \mathcal{G}}} \min_{\Xi(G,\Pi)} \max_{\Pi}PoA(G,\Xi, \Pi). $$

\paragraph{Graph Theory.} For every graph $G$, we denote by $V(G)$ and $E(G)$ 
the set of vertices and edges of $G$, respectively. For any $v,u \in V(G)$, 
$(v,u)$ denotes an edge between $v$ and $u$ and $d_G(v,u)$ denotes the shortest distance 
between $v$ and $u$ in $G$; if $G$ is clear from the context, we simply write $d(v,u)$. 
A graph $G$ is an {\em induced} 
subgraph of $H$, if $G$ is a subgraph of $H$ and for every $v,u \in V(G)$, $(v,u) \in E(G)$ 
if and only if $(v,u) \in E(H)$. 
$G$ is a {\em distance preserving} ({\em isometric}) 
subgraph of $H$, if $G$ is a subgraph of $H$ and for every $v,u \in V(G)$, $d_G(v,u) = d_H(v,u)$.


\section{Lower Bound of Ordered Protocols}
\label{generalLB}

The main result of this section is that the PoA of {\em any} ordered
protocol is $\Omega(\log k)$ which is tight. We formally define
(Definition~\ref{def:badOrder}) the `zig-zag' pattern of the lower
bounds of the Greedy Algorithm of the OSTP (see
Example~\ref{ex:order}(a) and Figure~\ref{fig:badPath}). Then the main
technical challenge is to show that {\em for any} ordering of the
vertices of high dimensional hypercubes, there always exists a {\em
  distance preserving} path, such that the order of its vertices
follows that zig-zag pattern. Finally, by connecting any two vertices
of the hypercube with a direct edge of suitable cost, similar to the
example in Figure\ref{fig:line_diam_sq}(a), we get the final lower
bound construction.

\begin{definition}[Classes]
\label{def:badOrderSets}

For $r>0$, and for a path graph $P=(v_0,\ldots, v_{2^{r}})$ of $2^r+1$
vertices, we define a partition of the vertices into $r+1$ {\em
  classes}, $D_0,D_1,\ldots, D_r$, as follows: 
Class $0$ contains the endpoints of $P$, $D_0 = \{v_0,v_{2^r}\}$. For
every $j\in [r]$, $D_j=\{v_i|\left(i \mod 2^{r-j}\right) = 0 \mbox{
  and } \left(i \mod 2^{r-j+1}\right) \neq 0\}$.  For $v \in D_j$, $w
\in D_{j'}$ and $ j < j'$, we say that $v$ belongs to a {\em lower}
class than $w$ (and $w$ belongs to a {\em higher} class than $v$).
\end{definition}

As an example, consider the path $P=
(v_0,v_1,v_2,v_3,v_4,v_5,v_6,v_7,v_8)$, where $r=3$. Then, $D_0 =
\{v_0,v_8\}$, $D_1 = \{v_4\}$, $D_2 = \{v_2,v_6\}$ and $D_3 =
\{v_1,v_3,v_5,v_7\}$.  Note that always $|D_0 |=2$ and for $j \neq 0$,
$|D_j|=2^{j-1}$. 

For $j > 0$ and $v_i \in D_j$, we define the {\em parents} of $v_i$ as
$\Pi(v_i) = \{ w| d_P(v_i,w)=2^{r-j}\}$, i.e. the closest vertices
that belong to lower classes. Remark that for all $v\notin \{v_0,
v_{2^{r}}\}$ i) the cardinality of $\Pi(v)$ is $2$, ii) the vertices
of $\Pi(v)$ belong to lower classes than $v$, iii) all vertices
between $v$ and any vertex of $\Pi(v)$ belong to higher classes than
$v$.  We are now ready to define the ``zig-zag" pattern.

\begin{definition}[Zig-zag pattern]
\label{def:badOrder}
We call a path graph $P=(v_0,v_1, \ldots, v_{2^{r}})$, with distinct integer labels $\pi$, {\em
  zig-zag}, and we denote it by $P_r(\pi)$, if for every $i \notin \{0,{2^{r}}\}$, $\pi(w) <\pi(v_i)$ for all $w \in \Pi(v_i)$.  
\end{definition}
An example of such a path for $r=3$ is shown in
Figure~\ref{fig:badPath}. Our main result of this section is that
there exist graphs, high dimensional hypercubes, such that for any
order $\pi$, $P_r(\pi)$ always appears as a {\em distance preserving}
subgraph. Our proof is existential and uses Ramsey theory.

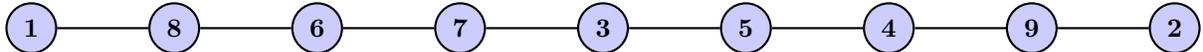
\begin{figure}[h]
\centering
\begin{tikzpicture}[xscale=1.9,yscale=0.2,decoration={markings, mark=between positions 0 and 1 step 10pt with { \draw [fill] (0,0) circle [radius=1.5pt];}}]
\node[draw,circle,thick,fill=blue!20,minimum width =0.5 cm] at (0,0) (v0)  {\small{$\mathbf{1}$}} ;
\node[draw,circle,thick,fill=blue!20,minimum width =0.5 cm] at (1,0) (v1)  {\small{$\mathbf{8}$}} ;
\node[draw,circle,thick,fill=blue!20,minimum width =0.5 cm] at (2,0) (v2)  {\small{$\mathbf{6}$}} ;
\node[draw,circle,thick,fill=blue!20,minimum width =0.5 cm] at (3,0) (v3)  {\small{$\mathbf{7}$}} ;
\node[draw,circle,thick,fill=blue!20,minimum width =0.5 cm] at (4,0) (v4)  {\small{$\mathbf{3}$}} ;
\node[draw,circle,thick,fill=blue!20,minimum width =0.5 cm] at (5,0) (v5)  {\small{$\mathbf{5}$}} ;
\node[draw,circle,thick,fill=blue!20,minimum width =0.5 cm] at (6,0) (v6)  {\small{$\mathbf{4}$}} ;
\node[draw,circle,thick,fill=blue!20,minimum width =0.5 cm] at (7,0) (v7)  {\small{$\mathbf{9}$}} ;
\node[draw,circle,thick,fill=blue!20,minimum width =0.5 cm] at (8,0) (v8)  {\small{$\mathbf{2}$}} ;
\path[draw, thick,-]
(v0) edge (v1)
(v1) edge (v2)
(v2) edge (v3)
(v3) edge (v4)
(v4) edge (v5)
(v5) edge (v6)
(v6) edge (v7)
(v7) edge (v8);
\end{tikzpicture}
\caption{An example of a $P_3(\pi)$ path. The numbers correspond to the labels.}
\label{fig:badPath}
\end{figure}

\vspace{5pt} \noindent{\bf Proof Overview:} The proof is by induction
and in the inductive step our starting point is the $n$-th dimensional
hypercube $Q_n$.  Given an ordering/labeling $\pi$ of the vertices of
$Q_n$ we first show that $Q_n$ contains a subgraph $W$ which is
isomorphic to a `pseudo-hypercube' $Q^2_m$ ($m<n$) where the labeling
of its vertices satisfies a special property (to be described
shortly).  $Q^2_m$ is defined by replacing each edge of $Q_m$ by a
2-edge path (of length two)\footnote{See $Q_m^2$ of
  Definition~\ref{def:Qn,s} and Figure~\ref{fig:Q2,4}a for an
  illustration}.

{\em Labeling property}: For the subgraph $W$ we require that all such
newly formed 2-edge paths, are $P_1(\pi)$ paths, i.e. the label
of the middle vertex is greater than the labels of the endpoints
(Figure~\ref{fig:Q2,4}(a) shows such a labeling).

Next, we contract all such 2-edge paths of $Q^2_m$ into single edges,
resulting in a graph isomorphic to $Q_m$; this is the hypercube used
for the next step. Note that each contracted edge still corresponds to
a path in $Q_n$. Therefore, after $r$ recursive steps, each edge
corresponds to a $2^r$ path of $Q_n$. Further, note that such a path
is a $P_r(\pi)$ path, due to the labeling property that we preserve at
each step. We require that, at the end of the last inductive step,
$Q_m=Q_1$ (a single edge), and (by unfolding it) we show that this edge
corresponds to a distance preserving subgraph of the original
graph/hupercube. At each step, $m<n$; the relation between $n$ and $m$
is determined by a Ramsey-type argument.
We next describe the basic ingredients that we use to show existence of
$W$. 
We apply a coloring scheme to the edges of $Q_n$ that depends on the
order of the vertices. 

{\em Coloring Scheme}: Consider $Q_n$ as a bipartite $Q_n=(A,B,E)$. 
For any edge $(v,u)$, with $v\in A$ and $u \in B$, if the $v$'s label 
is smaller than $u$'s, we paint the edge blue, otherwise we paint it red.

By a Ramsey-type argument we show that $Q_n$
has a {\em monochromatic} subgraph isomorphic to a specially defined
graph $G_m$; $G_m$ is carefully specified in such a way that it
contains at least two subgraphs isomorphic to pseudo-hypercubes
$Q^2_m$. The special property of those two subgraphs is described
next.

Let $H_1$ and $H_2$ be the two half cubes\footnote{The two half-cubes
  of order $n$ are formed from $Q_n$ by connecting all pairs of
  vertices with distance exactly two and dropping all other edges.}
of $Q_n$ and let $V(H_1) = A$ and $V(H_2)=B$.  Observe that if $Q^2_m$
is a subgraph of $Q_n$ then the corresponding $Q_m$ is an induced
subgraph of either $H_1$ or $H_2$. We carefully construct $G_m$ such
that it contains subgraphs $W_1$ and $W_2$ isomorphic to $Q^2_m$,
whose corresponding $Q_m$'s are induced subgraphs of $H_1$ and $H_2$,
respectively. The color of $G_m$ determines which of the $W_1$ and
$W_2$ will serve as the desired $W$. 
In particular, if the color is
blue, then for every edge $(v,u)$, with $v \in V(H_1)$ and $u \in
V(H_2)$, it should hold that $v$'s label 
is smaller than $u$'s and therefore the
labeling property is satisfied for $W_1$; similarly, if the color is
red, $W_2$ serves as $W$.
  
\vspace{5pt}

\noindent{\bf Proof Roadmap.} The whole proof of the lower bound
proceeds in several steps in the following sections.  In
Section~\ref{sec:Gm} we give the formal definition of the subgraph
$G_m$ of a hypercube $Q_n$.  Section~\ref{sec:ramsey} is devoted to
show that every $2$-edge coloring of a (suitably) high dimensional
hypercube contains a monochromatic copy of $G_m$
(Lemma~\ref{lem:monochrGm}), by using Ramsey theory.  Then, in
Section~\ref{sec:coloring} we show that, for any ordering of the
vertices of $Q_n$, we can define a special $2$-edge-coloring 
, so that there exists a $Q^2_m$
subgraph of $G_m$ that preserves the Labeling property
(Lemma~\ref{lem:isomToQ_m^2}).
At last, in Section~\ref{sec:LBconstruction}, by a recursive
application of the combination of the Ramsey-type result and the
coloring, we prove the existence of the {\em zig-zag} path in high
dimensional hypercubes (Theorem~\ref{thm:badPath}). We then show how
to construct a graph that serves as lower bound for all ordered
protocols (Theorem~\ref{thm:LBOrdered}). This is done by connecting
any two edges of the hypercube with a direct edge of appropriate cost,
similar to the example in
Figure~\ref{fig:line_diam_sq}(a).

\paragraph{Definitions and notation on Hypercubes.} We denote by
$[r,s]$ (for $r\leq s$) the set of integers $\{r,r+1,\ldots, s-1,s\}$,
but when $r=1$, we simply write $[s]$. We follow definitions and
notation of \cite{AlonRSV06}. Let $Q_n$ be the graph of the
$n$-dimensional hypercube whose vertex set is $\{0,1\}^n$. We
represent a vertex $v$ of $V(Q_n)$ by an $n$-bit string $x=\<
x_1\ldots x_n\> $, where $x_i\in \{0,1\}$.  By $\< xy\> $ or $xy$ we
denote the concatenation of an $r$-bit string $x$ with an $s$-bit
string $y$, i.e. $xy=\< x_1\ldots x_r y_1\ldots x_s\> $. $x=\< x_j\>
_{j=1}^r$ is the concatenation of its $r$ bits.  An edge is defined
between any two vertices that differ only in a single bit. We call
this bit, {\em flip-bit}, and we denote it by `$*$'. For example,
$x=\< 11100\> , y=\< 11000\> $ are two vertices of $Q_5$ and $(x,y)=\<
11*00\> $ is the edge that connects them. The distance between two
vertices $x,y$ is defined by their Hamming distance,
$d(v,u)=|\{j:x_j\neq y_j\}|$.  For a fixed subset of coordinates
$R\subseteq [n]$, we extend the definition of the distance as follows,
\begin{center}
$d(x,y,R) = \left\{
\begin{array}{l l}
d(x,y), &  \quad \mbox{if } \forall j\in [n] \setminus R,\; x_j=y_j  \\
\infty, &  \quad \mbox{otherwise.} \\
\end{array} \right.$ 
\end{center}

We define the {\em level} of a vertex $x$ by the number of `ones' it
contains, $\sum_{i=1}^nx_i$. We denote by $L_i$ the set of
vertices of level $i\in [0,n]$. We define the {\em prefix sum} of an
edge $e=(x,y)$, where the flip-bit is in the $j$-th coordinate, by
$p(e)=\sum_{i=1}^{j-1}x_i$. We represent any ordering $\pi$ of
$V(Q_n)$, by labeling the vertices with labels $1,\ldots, 2^n$, where
label $i$ corresponds to ranking $i$ in
$\pi$.  

\subsection{Description of $G_m$}
\label{sec:Gm}

For a positive integer $m$, we define a graph $G_m=(V_m,E_m)$ that is
a restriction of $Q_{4m}$ on $V_m=V_1\cup V_2 \cup V_3 \subseteq
V(Q_{4m})$.  A vertex of $V_1$ is defined by $2m-1$ concatenations of
pairs $\< 01\> $ and $\< 10\> $ and a single pair $\< 00\> $ that
appears in the {\em second} half of the string. A vertex of $V_2$ is
defined by $2m$ concatenations of $\< 01\> $ and $\< 10\> $.  A vertex
of $V_3$ is defined by $2m-2$ concatenations of $\< 01\> $ and $\<
10\> $, one pair $\< 11\> $ that appears on the first half of the
string, and one pair $\< 00\> $ that appears on the second half.  For
example, for $m=2$, $\< 01\;10\;00\;10\> \in V_1$, $\<
01\;10\;10\;10\> \in V_2$, $\< 01\;11\;10\;00\> \in V_3$.  More
formally, let $A = \{\< 01\> ,\< 10\> \}$, then the subsets $V_1, V_2,
V_3$ are defined as follows:
\begin{eqnarray*}
V_1:=V_1(m) =& \{\< a_j b_j\> _{j=1}^{2m}|& \exists i \in [m+1,2m] \mbox{ s.t. } \< a_i b_i\>  = \< 00\>  
\mbox{ and } \forall j\neq i, \< a_j b_j\>  \in A  \},\\
V_2:=V_2(m) =& \{\< a_j b_j\> _{j=1}^{2m}|&\forall j, \< a_j b_j\>  \in A \},\\
V_3:=V_3(m) = &\{\< a_j b_j\> _{j=1}^{2m}| & \exists i_1 \in [m], \exists i_2 \in [m+1,2m] \mbox{ s.t. } \\ 
&& \< a_{i_1} b_{i_1}\>  = \< 11\> , \;  \< a_{i_2} b_{i_2}\>  = \< 00\>  \mbox{ and } \forall j\neq i_1,i_2, \< a_j b_j\>  \in A\}.
\end{eqnarray*}
Observe that $G_m$ is bipartite with vertex partitions $V_1$ and
$V_2\cup V_3$, as vertices of $V_1$ belong to level $2m-1$, while vertices of
$V_2\cup V_3$ to level $2m$.

\begin{lemma}
\label{lem:1ConnectionofHd}
Every pair of vertices $x,x' \in V_1(m)$ with $d(x,x',[2m])=2$, have a
unique common neighbor $y \in V_3(m)$. Also, every pair of vertices $x,x'
\in V_2(m)$, with $d(x,x',[2m+1,4m])=2$, have a unique common neighbor $y \in
V_1(m)$.
\end{lemma}

\begin{proof}
  Recall that (by definition) if $d(x,x',R)\neq \infty$ then $x,x'$ should coincide in
  {\em all but} the $R$ coordinates.  For the first statement, observe
  that the premises of the Lemma hold only if there exists $s \in [m]$
  such that $x_{2s-1}x_{2s}=\< 10\> $ and $x'_{2s-1}x'_{2s}=\< 01\> $
  (or the other way around), in which case the required vertex $y$
  from $V_3(m)$ has $y_{2s-1}y_{2s}=\< 11\> $; the rest of the bits
  are the same among $x,x',y$.  For the second statement, the premises
  of the Lemma hold only if there exists an $s \in [m+1,2m]$ such that
	$x_{2s-1}x_{2s}=\< 10\> $ and $x'_{2s-1}x'_{2s}=\< 01\> $ (or the
  other way around), in which case the required vertex $y$ from $V_1(m)$
  has $y_{2s-1}y_{2s}=\< 00\> $ and the rest of the bits are the same among $x,x',y$.
\end{proof}

\subsection{Ramsey-type Theorem}
\label{sec:ramsey}
\begin{lemma}
\label{lem:monochrGm}
For any positive integer $m$, and for sufficiently large $n \geq
n_0=g(m)$, any 2-edge coloring $\chi$ of $Q_n$, contains a
monochromatic copy of $G_{m}$\footnote{The result could be extended to
  any (fixed) number of colors, but we need only two for our application.}.
\end{lemma}

\begin{proof}
  The proof follows ideas of Alon et al.~\cite{AlonRSV06}. Consider a
  hypercube $Q_n$, with sufficiently large $n > 6m$ to be determined
  later, and some arbitrary $2$-edge-coloring
  $\chi:E(Q_n)\rightarrow\{1,2\}$. Let $E^*$ be the set of edges
  between vertices of $L_{4m-1}$ and $L_{4m}$ (recall that $L_i = \{
  v| w(v) = i\}$).

Each edge $e \in E^*$ contains $4m-1$ 1's, a flip-bit represented by $*$ and 
the rest of the coordinates are $0$. Moreover, $e$ is uniquely
determined by its $4m$ non-zero coordinates $R_e\subseteq [n]$ 
and its prefix sum $p(e)\in [0,4m-1]$ (number of $1'$s before the flip-bit). 
Therefore, the color $\chi(e)$ defines a coloring of the pair
$(R_e,p(e))$, i.e. $\chi(e) = \chi(R_e,p(e))$. For each subset
$R\subset [n]$ of
$4m$ coordinates, we denote by $c(R)=(\chi(R,0),...,\chi(R,4m-1))$ the
color induced by the edge coloring. The coloring of all subsets $R$
defines a coloring of the complete $4m$-uniform hypergraph of $[n]$
using $2^{4m}$ colors.

By Ramsey's Theorem for hypergraphs~\cite{Ram90}, there exists
$n_0=g(m)$ such that for any $n\geq n_0$ there exists some subset $U
\subset [n]$ of size $6m$ such that all $4m$-subsets $R \subset U$
have the same color $c(R) = c^*$. Therefore, for every $R_1, R_2
\subset U$ and $p \in [0,4m-1]$, it is
$\chi(R_1,p)=\chi(R_2,p)=c_p$. Since $p$ takes $4m$ values and there
are only two different colors, there must exist $2m$ indices
$p_0,\ldots, p_{2m-1} \in [0,4m-1]$ with the same color $\chi(R,p_i)
= \chi^*$, for all $R \subset U$, $|R| = 4m$ and $i \in [0,2m-1]$.

It remains to show that the graph formed with the edges that are
determined by those prefix sums, contains a monochromatic copy of
$G_m$. We will show this by constructing those edges from $E_{m}$ (the
set of edges of $G_{m}$). By inserting blocks of $1$'s
of suitable length among the bits of the edges of $E_m$, we 
construct the bits at the coordinates of $U$. The rest of the bits ($n-|U|$) are set to
zero.

Let $1^{r}$ be a string of $r$ 1's and define $\beta_i =
1^{p_i-p_{i-1}-1}$ for $i \in [2m-1]$, $\beta_0=1^{p_0}$ and
$\beta_{2m} = 1^{4m-1-p_{2m-1}}$. For any edge $e = \< a_jb_j\> _j \in
E_m$, we insert $\beta _0$ at the beginning of the string, for $j \in
[m]$ we insert $\beta_{j}$ between $a_j$ and $b_j$ and for $j \in
[m+1,2m]$ we insert the string $\beta_{j}$ after $b_j$.
Recall that each edge of $E_m$ contains exactly $2m$ zero bits.  Also
notice that $\sum_j |\beta_j| = p_0 +
\sum_{i=1}^{2m-1}\left(p_i-p_{i-1}-1\right)+ 4m-1-p_{2m-1} =
-(2m-1)+4m-1=2m.$ Therefore, in total we have $6m$ bits (same as the
size of $U$) and $4m$ non-zero bits (same as the size of $R$). These
$6m$ bits are put precisely at the coordinates of $U$. The rest $n-6m$
of the coordinates are filled with zeros.

It remains to show that for such edges the prefix of the flip-bit is
{\em always} one of the $p_0, \ldots , p_{2m-1}$. This would imply
that all these edges are monochromatic. Furthermore, all but $4m$
coordinates are fixed and the $4m$ coordinates form exactly the sets
$V_1(m), V_2(m), V_3(m)$; therefore, the monochromatic subgraph is
isomorphic to $G_m$.

For any edge $e = \< a_jb_j\> _j \in E_m$, let the flip-bit be at position: 
\begin{itemize}
\item $a_j$ for $j \in [m]$. Its prefix is $\sum_{i=0}^{j-1}\beta_i +
  (j-1) = p_{j-1}$, where the term $j-1$ corresponds to the number of
  pairs $\< a_sb_s\> $ with $s < j$, each of which contributes to the
  prefix with a single $1$.
\item $b_j$ for $j \in [m]$. Since $j \leq m$, $a_j = 1$. Then the
  prefix equals to $\sum_{i=0}^{j}\beta_i + (j-1) +1 = p_j$.
\item $a_j$ or $b_j$ for $j \in [m+1,2m]$. For such $j$, $\< a_jb_j\> 
  \in  \{ \< 0* \>  , \< *0 \> \}$  and all other pairs belong to $A$.
  Therefore, the prefix is equal to $\sum_{i=0}^{j-1}\beta_i + (j-1) =
  p_{j-1}$.
\end{itemize}
\end{proof}

\subsection{Coloring based on the labels}
\label{sec:coloring}
This part of the proof shows that for any ordering of the vertices of
a hypercube $Q_n$, there is a $2$-edge coloring with the following
property: {\em in the monochromatic $G_m$, either all the vertices of $V_1$ 
  or all the vertices of $V_2$  have neighbors in $G_m$ with only higher
  label.}  This implies a desired labeling property for a $Q_m^2$
subgraph of $Q_n$, the structure of which is defined next.

\begin{definition}
\label{def:Qn,s}
We define $Q_n^s$ to be a subdivision of $Q_n$, by replacing each edge
by a path of length $s$. $Q_n^1$ is simply $Q_n$.
We denote by $Z(Q^s_n)$ the set of all pairs of vertices $(x,x')$,
which correspond to edges of $Q_n$; $P(x,x')$ is the corresponding
path in $Q^s_n$. For every $(x,x') \in Z(Q_m^2)$, we denote by
$\theta(x,x')$ the middle vertex of $P(x,x')$.

\begin{figure}[h]
\centering
\begin{tabular}{c c}

\begin{tikzpicture}[xscale=0.25,yscale=0.25,decoration={markings, mark=between positions 0 and 1 step 10pt with { \draw [fill] (0,0) circle [radius=1.5pt];}}]

\node[draw,circle,thick,fill=black,inner sep=1.4pt,radius=0.25pt,label={[label distance=-4pt]below left:\scriptsize{$1$}}] at (2,0) (v1) {} ;
\node[draw,circle,thick,fill=black,inner sep=1.4pt,radius=0.25pt,label={[label distance=-4pt]below left:\scriptsize{$4$}}] at (2,8) (v2) {} ;
\node[draw,circle,thick,fill=black,inner sep=1.4pt,radius=0.25pt,label={[label distance=-4pt]below right:\scriptsize{$12$}}] at (10,8) (v3) {} ;
\node[draw,circle,thick,fill=black,inner sep=1.4pt,radius=0.25pt,label={[label distance=-4pt]below right:\scriptsize{$8$}}] at (10,0) (v4) {};
\node[draw,circle,thick,fill=black,inner sep=1.4pt,radius=0.25pt,label={[label distance=-4pt]above left:\scriptsize{$11$}}] at (0,2) (v5) {} ;
\node[draw,circle,thick,fill=black,inner sep=1.4pt,radius=0.25pt,label={[label distance=-4pt]above left:\scriptsize{$2$}}] at (0,10) (v6) {} ;
\node[draw,circle,thick,fill=black,inner sep=1.4pt,radius=0.25pt,label={[label distance=-4pt]above right:\scriptsize{$5$}}] at (8,10) (v7) {} ;
\node[draw,circle,thick,fill=black,inner sep=1.4pt,radius=0.25pt,label={[label distance=-4pt]above right:\scriptsize{$3$}}] at (8,2) (v8) {} ;

\node[draw,circle,thick,fill=blue,inner sep=1.4pt,radius=0.25pt,label={[label distance=-2pt]left:\scriptsize{$19$}}] at (2,4) (v12) {} ;
\node[draw,circle,thick,fill=blue,inner sep=1.4pt,radius=0.25pt,label={[label distance=-2pt]below:\scriptsize{$13$}}] at (6,8) (v23) {} ;
\node[draw,circle,thick,fill=blue,inner sep=1.4pt,radius=0.25pt,label={[label distance=-2pt]right:\scriptsize{$20$}}] at (10,4) (v34) {} ;
\node[draw,circle,thick,fill=blue,inner sep=1.4pt,radius=0.25pt,label={[label distance=-2pt]below:\scriptsize{$10$}}] at (6,0) (v41) {} ;
\node[draw,circle,thick,fill=blue,inner sep=1.4pt,radius=0.25pt,label={[label distance=-4pt]below left:\scriptsize{$15$}}] at (1,1) (v15) {} ;
\node[draw,circle,thick,fill=blue,inner sep=1.4pt,radius=0.25pt,label={[label distance=-5pt]below left:\scriptsize{$6$}}] at (1,9) (v26) {} ;
\node[draw,circle,thick,fill=blue,inner sep=1.4pt,radius=0.25pt,label={[label distance=-4pt]above right:\scriptsize{$14$}}] at (9,9) (v37) {};
\node[draw,circle,thick,fill=blue,inner sep=1.4pt,radius=0.25pt,label={[label distance=-5pt]above right:\scriptsize{$9$}}] at (9,1) (v48) {};
\node[draw,circle,thick,fill=blue,inner sep=1.4pt,radius=0.25pt,label={[label distance=-2pt]left:\scriptsize{$18$}}] at (0,6) (v56) {};
\node[draw,circle,thick,fill=blue,inner sep=1.4pt,radius=0.25pt,label={[label distance=-2pt]above:\scriptsize{$17$}}] at (4,10) (v67) {};
\node[draw,circle,thick,fill=blue,inner sep=1.4pt,radius=0.25pt,label={[label distance=-2pt]right:\scriptsize{$7$}}] at (8,6) (v78) {};
\node[draw,circle,thick,fill=blue,inner sep=1.4pt,radius=0.25pt,label={[label distance=-2pt]above:\scriptsize{$16$}}] at (4,2) (v85) {} ;

\path[postaction={decorate}] (2.5,0) to (3.5,0);

\path[draw, thick,-]
(v12) edge (v1)
(v12) edge (v2)
(v23) edge (v2)
(v23) edge (v3)
(v34) edge (v3)
(v34) edge (v4)
(v41) edge (v4)
(v41) edge (v1)
(v15) edge (v1)
(v15) edge (v5)
(v26) edge (v2)
(v26) edge (v6)
(v37) edge (v3)
(v37) edge (v7)
(v48) edge (v4)
(v48) edge (v8)
(v56) edge (v5)
(v56) edge (v6)
(v67) edge (v6)
(v67) edge (v7)
(v78) edge (v7)
(v78) edge (v8)
(v85) edge (v8)
(v85) edge (v5);
\end{tikzpicture}
\qquad\qquad
&
\qquad\qquad
\begin{tikzpicture}[xscale=0.3,yscale=0.3,decoration={markings, mark=between positions 0 and 1 step 10pt with { \draw [fill] (0,0) circle [radius=1.5pt];}}]

\node[draw,circle,thick,fill=black,inner sep=1.4pt,radius=0.5pt,label={[label distance=-2pt]left:\scriptsize{$1$}}] at (0,0) (v0) {} ;
\node[draw,circle,thick,fill=blue,inner sep=1.4pt,radius=0.5pt,label={[label distance=-2pt]left:\scriptsize{$12$}}] at (0,2) (v1) {} ;
\node[draw,circle,thick,fill=blue,inner sep=1.4pt,radius=0.5pt,label={[label distance=-2pt]left:\scriptsize{$5$}}] at (0,4) (v2) {} ;
\node[draw,circle,thick,fill=blue,inner sep=1.4pt,radius=0.5pt,label={[label distance=-2pt]left:\scriptsize{$6$}}] at (0,6) (v3) {} ;
\node[draw,circle,thick,fill=black,inner sep=1.4pt,radius=0.5pt,label={[label distance=-2pt]left:\scriptsize{$2$}}] at (0,8) (v4) {} ;
\node[draw,circle,thick,fill=blue,inner sep=1.4pt,radius=0.5pt,label={[label distance=-2pt]below:\scriptsize{$7$}}] at (2,8) (v5) {} ;
\node[draw,circle,thick,fill=blue,inner sep=1.4pt,radius=0.5pt,label={[label distance=-2pt]below:\scriptsize{$4$}}] at (4,8) (v6) {} ;
\node[draw,circle,thick,fill=blue,inner sep=1.4pt,radius=0.5pt,label={[label distance=-2pt]below:\scriptsize{$9$}}] at (6,8) (v7) {} ;
\node[draw,circle,thick,fill=black,inner sep=1.4pt,radius=0.5pt,label={[label distance=-2pt]right:\scriptsize{$3$}}] at (8,8) (v8) {} ;
\node[draw,circle,thick,fill=blue,inner sep=1.4pt,radius=0.5pt,label={[label distance=-2pt]right:\scriptsize{$14$}}] at (8,6) (v9) {} ;
\node[draw,circle,thick,fill=blue,inner sep=1.4pt,radius=0.5pt,label={[label distance=-2pt]right:\scriptsize{$13$}}] at (8,4) (v10) {};
\node[draw,circle,thick,fill=blue,inner sep=1.4pt,radius=0.5pt,label={[label distance=-2pt]right:\scriptsize{$16$}}] at (8,2) (v11) {};
\node[draw,circle,thick,fill=black,inner sep=1.4pt,radius=0.5pt,label={[label distance=-2pt]right:\scriptsize{$8$}}] at (8,0) (v12) {};
\node[draw,circle,thick,fill=blue,inner sep=1.4pt,radius=0.5pt,label={[label distance=-2pt]above:\scriptsize{$15$}}] at (6,0) (v13) {};
\node[draw,circle,thick,fill=blue,inner sep=1.4pt,radius=0.5pt,label={[label distance=-2pt]above:\scriptsize{$10$}}] at (4,0) (v14) {};
\node[draw,circle,thick,fill=blue,inner sep=1.4pt,radius=0.5pt,label={[label distance=-2pt]above:\scriptsize{$11$}}] at (2,0) (v15) {};


\path[draw, thick,-]
(v0) edge (v1)
(v1) edge (v2)
(v2) edge (v3)
(v3) edge (v4)
(v4) edge (v5)
(v5) edge (v6)
(v6) edge (v7)
(v7) edge (v8)
(v8) edge (v9)
(v9) edge (v10)
(v10) edge (v11)
(v11) edge (v12)
(v12) edge (v13)
(v13) edge (v14)
(v14) edge (v15)
(v15) edge (v0);

\end{tikzpicture}\\
(a)\qquad\qquad & \qquad\qquad(b)

\end{tabular}
\caption{Examples of (a) $Q_3^2$ and (b) $Q_2^4$. The labels on the nodes are examples of the labeling 
property, (a) after one inductive step, (b) after two inductive steps.}
\label{fig:Q2,4}
\end{figure}
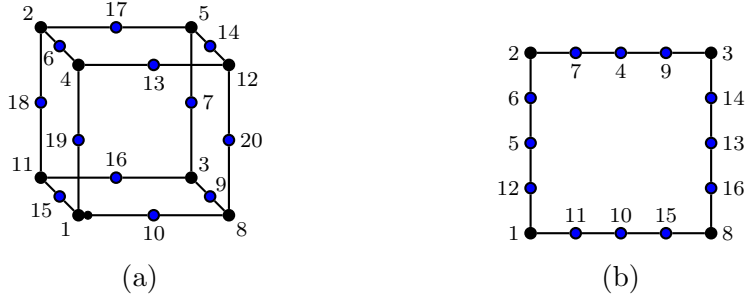 

\end{definition}

In the next lemma we show that for any ordering of the vertices of
$Q_n$, there exists a subgraph isomorphic to $Q_m^2$, such that the `middle'
vertices have higher label than their neighbors (Labeling Property). 

\begin{lemma}
\label{lem:isomToQ_m^2}
For any positive integer $m$, for all $n\geq n_0=g(m)$ and for any
ordering $\pi$ of $V(Q_n)$, there exists a subgraph $W$ of $Q_n$
that is isomorphic to $Q_m^2$, such that for every $(x,x') \in Z(W)$, 
it is $\pi(\theta(x,x')) >\max\{ \pi(x),\pi(x')\}$.
\end{lemma}

\begin{proof}
  Choose a sufficiently large $n\geq n_0=g(m)$ as in
  Lemma~\ref{lem:monochrGm}. Partition the vertices of $Q_{n}$ into
  sets $\mathcal{O}, \mathcal{E}$ of vertices of odd and even level,
  respectively.
  We color the edges of $Q_n$ as follows. For every edge $e=(z,z')$
  with $z\in \mathcal{O}$ and $z'\in \mathcal{E}$, if $\pi(z) <
  \pi(z')$, then paint $e$ blue. Otherwise paint it red. Therefore,
  for every blue edge, the endpoint in $\mathcal{O}$ has smaller label
  than the endpoint in $\mathcal{E}$. The opposite holds for any red
  edge.

  Lemma~\ref{lem:monochrGm} implies that $Q_n$ contains a
  monochromatic copy (blue or red) of $G_{m}$. Recall that $G_m$ is
  bipartite between vertices of levels $L_{4m-1}$ and $L_{4m}$ and
  that $V_1 \subset L_{4m-1} \subset \mathcal{O} $ and $V_2 \cup V_3
  \subset L_{4m} \subset \mathcal{E}$.
  Let $R\subset [n]$ be the subset of the $4m$ coordinates that
  correspond to vertices of $G_m$. Also let $R_1$ and $R_2$ be the
  subsets of the first $2m$ and the last $2m$ coordinates of $R$,
  respectively.

  First suppose that $G_m$ is blue. An immediate implication of our
  coloring is that for every edge $(z,z') \in E_m$ with $z \in V_1$,
  $z' \in V_2 \cup V_3$ it must be $\pi(z) < \pi(z')$. Fix a $2m$-bit
  string $s$ that corresponds to a permissible bit assignment to the
  $R_2$ coordinates of some vertex in $V_1$ (see Section
  \ref{sec:Gm}). Define $W_s$ as the subset of vertices of $V_1$ where
  the $R_2$ coordinates are set to $s$.
  Recall that each of the first $m$ pairs $\< a_jb_j\> $, $j \in [m]$,
  of a vertex $z\in W_s$, may take any of the two bit assignments $\<
  01\> $ and $\< 10\> $. Hence, $|W_s|=2^m$.

  Observe that we can embed $W_s$ into $Q_m$ with distortion 1 and
  scaling factor $1/2$, by mapping the first $m$ pairs of bits into
  single bits; map $\< 01\> $ to $0$ and $\< 10\> $ to $1$. Every two
  vertices with distance $d$ in $Q_m$, have distance $2d$ in $Q_n$.
For every $x,x' \in W_s\subset V_1$ with $d(x,x')=2$,  
Lemma~\ref{lem:1ConnectionofHd} implies that there exists
$y=\theta(x,x')\in V_3$, such that $d(x,y) = d(x',y) = 1$. Therefore,
$\pi(y) > \max\{\pi(x),\pi(x')\}$. Take the union $Y=\cup_y$ of
all such vertices $y$, then $W_s\cup Y$ induces a subgraph $W$
isomorphic to $Q_m^2$, that fulfills the labeling requirements.

The case of $G_m$ being red is similar. We focus only on vertices
$V_2$. Fix now a $2m$-bit string $s$ that corresponds to a
permissible bit assignment of the $R_1$ coordinates of a vertex in
$V_2$.  Define $W_s$ as the subset of vertices of $V_2$ where the
$R_1$ coordinates are set to $s$. Similarly, we can embed $W_s$ into
$Q_m$ with distortion 1 and scaling factor $1/2$.

For every $x,x' \in W_s\subset V_2$ with $d(x,x')=2$, where the $R_1$
coordinates are fixed to $s$, Lemma~\ref{lem:1ConnectionofHd} implies
that there exists $y=\theta(x,x')\in V_1$, such that $d(x,y) = d(x',y)
= 1$. Therefore, $\pi(y) > \max\{\pi(x),\pi(x')\}$. Take the union
$Y=\cup_y$ of all such vertices $y$, then $W_s\cup Y$ induces a
subgraph $W$ isomorphic to $Q_m^2$, that fulfills the labeling
requirements.
\end{proof}

\subsection{Lower Bound Construction}
\label{sec:LBconstruction}
Now we are ready to prove the main theorem of this section. 

\begin{theorem}
\label{thm:badPath}
For every positive integer $r$, and for sufficiently large $n\geq
n(r)$, there exists a graph $Q_n$ such that, for {\em every} ordering $\pi$
of its vertices, it contains a zig-zag {\em distance preserving} path
$P_r(\pi)$.
\end{theorem}

\begin{proof}
  Let $g$ be a function as in Lemma~\ref{lem:monochrGm}. We
  recursively define the sequence $n_0, n_1, \ldots, n_r$, such that
  $n_r = 1$ and $n_{i-1} = g(n_{i})$, for $i \in [r]$.   
  We will show that $Q_{n_0}$ is the graph we are looking for.

\begin{claim}
\label{cl:mathInduction}
For every $i \in [0,r]$, and for any vertex ordering $\pi$ of
$Q_{n_0}$, it contains a subgraph isomorphic to $Q_{n_i}^{2^i}$, such
that for every $(x,x') \in Z(Q_{n_i}^{2^i})$, $P(x,x')$ is a distance 
preserving path isomorphic to $P_i(\pi)$.
\end{claim}

\begin{proof}
  The proof is by induction on $i$. As a base case, $Q_{n_0}^{2^0} =
  Q_{n_0}$ is the graph itself. An edge is trivially a path
  $P_0(\pi)$, for any $\pi$.
  Suppose now that $Q_{n_0}$ contains a subgraph isomorphic to
  $Q_{n_i}^{2^i}$, for some $i<r$, such that for every $q \in
  Z(Q_{n_i}^{2^i})$, $P(q)$ is a path $P_i(\pi)$.  It is sufficient to
  show that $Q_{n_i}^{2^i}$ contains a subgraph isomorphic to
  $Q_{n_{i+1}}^{2^{i+1}}$, such that for every $q \in
  Z(Q_{n_{i+1}}^{2^{i+1}})$, $P(q)$ is a path $P_{i+1}(\pi)$.

  For every $(x,x')\in Z(Q_{n_i}^{2^i})$, if we replace $P(x,x')$ with
  a direct edge $e=(x,x')$, the resulting graph is a copy of
  $Q_{n_i}$.  Applying Lemma~\ref{lem:isomToQ_m^2} on $Q_{n_i}$,
  guarantees the existence of a subgraph $W$ isomorphic to
  $Q_{n_{i+1}}^2$ ($n_i = g(n_{i+1})$), where for every $(y,y') \in
  Z(W)$, $\pi(\theta(y,y')) > \max\{\pi(y),\pi(y')\}$. Each of the
  edges $(y,\theta(y,y'))$ and $(y',\theta(y,y'))$ of $Q_{n_{i+1}}^2$
  are replaced by a path $P_i(\pi)$ in $Q_{n_i}^{2^i}$. Therefore, $W$
  is a copy of $Q_{n_{i+1}}^{2^{i+1}}$, with $P(y,y')$ being a path
  $P_{i+1}(\pi)$.
\end{proof}

We now argue that the resulting $P_r(\pi)$ is a distance preserving
path.  Our analysis indicate a sequence of hypercubes $Q_{n_0},
Q_{n_1}, \ldots, Q_{n_r}$.  Recall that in Lemma
\ref{lem:isomToQ_m^2}, in order to get $Q_{n_{i+1}}$ from $Q_{n_{i}}$
we mapped $\< 01\> $ to $0$ and $\< 10\> $ to $1$ and the vertices of
$Q_{n_{i+1}}$ did not differ in any other bit but the ones we mapped.
Consider now the two vertices $x,x'$ of $Q_{n_r}=Q_1$ with bit-strings
$\< 0\> $ and $\< 1\> $, respectively. Their Hamming distance in their
original bit representation (in $Q_{n_0}$) should be $2^{r}$, the same
with their distance in $P_{r}(\pi)$. Moreover, if any two vertices of
$P_{r}(\pi)$ are closer in $Q_{n_0}$ than in $P_{r}(\pi)$, then this
would contradict the fact that $d_{Q_{n_0}}(x,x')=2^r$.
\end{proof}

Finally we extend $Q_n$ so that for any order $\pi$ of its vertices, a path $P_r(\pi)$ exists 
along with the shortcuts as shown in the example in Figure~\ref{fig:line_diam_sq}(a).

\begin{theorem}
\label{thm:LBOrdered}
Any ordered universal 
cost-sharing protocol on undirected graphs admits a PoA of
$\Omega(\log k)$, where $k$ is the number of activated vertices.
\end{theorem}

\begin{proof}
  Let $k=2^r+1$ for some positive integer $r$. From
  Theorem~\ref{thm:badPath}, we know that for any vertex ordering
  $\pi$ of $Q_{n(r)}$ there is a distance preserving path $P_r(\pi)$.

  We use $Q_{n(r)}$ as a basis to construct the weighted graph
  $\tilde{Q}_{n(r)}$ with vertex set $V(\tilde{Q}_{n(r)})=Q_{n(r)}\cup
  \{t\}$, where $t$ is the designated root. We connect every pair of
  vertices $x,y$ with a direct edge of cost $c_e=2^k$, if $t$ is one
  of its endpoints, otherwise its cost is $c_e=d_{Q_{n(r)}}(x,y)$
  (similar to Figure~\ref{fig:line_diam_sq}(a)).

  The adversary selects to activate the vertices of $P_r(\pi)$, and
  the lower bound follows; in the NE the players choose their direct
  edges to connect with one of their parents (see
  at the beginning of Section \ref{generalLB} for the term ``parent").
\end{proof}


\section{Lower Bound for all universal protocols}

In this section, we exhibit metric spaces for which no universal
cost-sharing protocol admits a PoA better than $\Omega (\log k)$.  Due
to the characterization of \cite{CRV10}, we can restrict ourselves in
generalized weighted Shapley protocols (GWSPs). We follow the notation
of \cite{CRV10}, and for the sake of self-containment we include here
the most related definitions and lemmas.

\subsection{Cost-Sharing Preliminaries}
A strictly positive function $f:2^{N} \rightarrow \mathbb{R}^{+}$ is
an {\em edge potential} on $N$, if it is strictly increasing, i.e. for
every $R \subset S \subseteq N$, $f(R) < f(S)$, and for every $S
\subseteq N$, $\sum_{i \in S} \frac{f(S)-f(S\setminus
  \{i\})}{f(\{i\})}=1.$ 
For simplicity, instead of $f(\{i\})$, we write $f(i)$. A cost-sharing
protocol is called {\em potential-based}, if it is defined by
assigning to each edge of cost $c$, the cost-sharing method $\xi$,
where for every $S \subseteq N$ and $i \in S$, $\xi(i,S) = c\cdot
\frac{f(S)-f(S\setminus \{i\})}{f({i})}.$

Let $\Xi_1$ and $\Xi_2$ be two cost-sharing protocols for {\em disjoint} sets of vertices  
$U_1$ and $U_2$, with methods $\xi_1$ and $\xi_2$, respectively. 
The {\em concatenation} of $\Xi_1$ and $\Xi_2$ 
is the cost sharing protocol $\Xi$ 
of the set $U_1 \cup U_2$, with method $\xi$ defined as  
	\begin{center}
  \begin{tabular}{l}
$\xi(i,S) = \left\{
\begin{array}{l l}
\xi_1(i,S \cap U_1) & \text{if } \;\; i \in U_1  \\
\xi_2(i,S) & \mbox{if } \;\; S \subseteq U_2\\
0 & \mbox{otherwise} \\
\end{array} \right.$ 
  \end{tabular}
\end{center}

Note that the concatenation of two protocols for disjoint sets of
vertices defines an order among these two sets. The GWSPs are
concatenations of potential-based protocols.

\begin{lemma} 
\label{lem:proximity}
(Lemma $4.10$ of \cite{CRV10}). Let $f$ be an edge potential on $N$ and 
$\xi$ the induced (by $f$) cost-sharing method, for unit costs. 
For $k \geq 1$ and a constant $\alpha$, with $1 \leq \alpha^{2k} \leq 1+k^{-3}$, 
let $S \subseteq N$ be a subset of vertices with $f(i) \leq \alpha f(j)$, 
for every $i,j \in S$. If $|S| \leq k$, then for any $i,j \in S$, 
$\xi(i,S) \leq \alpha (\xi(j,S)+2k^{-2}).$ 
\end{lemma}

\begin{lemma}  
\label{lem:separation}
(Lemma $4.11$ of \cite{CRV10}). Let $f$ be an edge potential on $N$,
and $\xi$ be the cost-sharing method induced by $f$, for unit cost.
For any two vertices $i,j \in N$, such that $f(i) \geq \beta f(j)$:
$\xi (i,\{i,j\}) \geq \beta/(\beta+1)$ and for every
$S\supseteq\{i,j\}$, $\xi (j,S) \leq 1/(\beta+1).$
\end{lemma}

\subsection{Lower Bound}
The following two technical lemmas will be used in our main theorem.

\begin{lemma}
\label{lem:matching-s}
Let $X$ be a finite set of size $msr^2$, and $X_1, \ldots, X_m$ be a
partition of $X$, with $|X_i|=sr^2$, for all $i\in [m]$. Then, for any
coloring $\chi$ of $X$ such that {\em no more} than $r$ elements have
the same color, there exists a {\em rainbow} subset $S \subset X$
(i.e. $\chi(v)\neq\chi(u)$ for all $v,u \in S$),
with $|S \cap X_i|=s$ for every $i \in [m]$.
\end{lemma}

\begin{proof}
  Given the partition $X_1,\ldots, X_m$ of $X$ and the coloring
  $\chi$, we construct a bipartite graph $G=(A,B,E)$, where $A$ is the
  set of colors used in $\chi$. For every $X_i$ we create a set $B_i$
  of size $s$; then $B=\cup B_i$. If color $j$ is used in $X_i$,
  we add an edge $(j,l)$ for all $l\in B_i$.

Each color $j \in A$ appears in at most $r$ distinct $X_i$ sets, and since for each
$X_i$ there are $s$ vertices ($B_i$), the degree of $j$ is 
at most $rs$. On the other hand, each $X_i$ has size $r^2 s$
and hence, it has at least $rs$ different colors. Therefore, the degree
of each vertex of $B$ is at least $rs$.

Consider any set $R\subseteq B$, and let $E(R)$ be the set of edges with at least
one endpoint in $R$. If $N(R)$ denotes
the set of neighbors of $R$, observe that $E(R) \subseteq E(N(R))$. 
By using the degree bound on vertices of 
$B$, $|E(R)| \geq rs |R|$ and by using the
degree bound on vertices of $A$, $|E(N(R))| \leq rs|N(R)|$. 
Therefore, $|R| \leq |N(R)|$.  By Hall's Theorem there exists
a matching which covers every vertex in $B$. 
Each vertex in $B_i$ is matched with a distinct color and therefore in each 
$X_i$ there exists a subset with at least
$s$ elements with distinct colors; let $W_i$ be such a subset with exactly $s$ elements. 
In addition the colors in different $W_i$ subsets should be distinct by the matching. 
Then, $S=\cup W_i$.
\end{proof}

\begin{lemma}
\label{lem:nonconsequent}
Let $X = (X_1,\ldots, X_m)$  be a partition of $[m^2]$, with $|X_i|=m$, for all $i \in
[m]$. Then, there exists a subset $S \subset [m^2]$ with exactly one
element from each subset $X_i$, such no two distinct
$x,y\in S$ are consecutive, i.e. for every $x,y \in S$, $|x - y| \geq 2$.
\end{lemma}

\begin{proof}
  For every $i$, let $X_i=\{x_{i1}, \ldots, x_{im}\}$. W.l.o.g we can
  assume that the $x_{ij}$'s are in increasing order with respect to
  $j$ and in addition that $X_i$'s are sorted such that $x_{ii} <
  x_{ji}$, for all $j>i$ (otherwise rename the elements recursively to
  fulfill the requirement). Then, it is not hard to see that $S =
  \{x_{kk}|k\in [m]\}$ can serve as the required set.
\end{proof}

Now we proceed with the main theorem of this section. We create a
graph where every GWSP has high PoA. At a high level, we construct a
high dimensional hypercube with sufficiently large number of potential
players at each vertex (by adding many copies of each vertex connected
via zero-cost edges). Moreover, we add shortcuts among the vertices of
suitable costs and we connect each vertex with $t$ via two parallel
links with costs that differ by a large factor (see Figure
\ref{fig:allProtConstr}). If the protocol induces a large enough set
of potential players with {\em Shapley-like} values in some vertex,
then it is a NE that all these players follow the most costly link to
$t$. Otherwise, by using Lemmas~\ref{lem:matching-s}
and~\ref{lem:nonconsequent} we show that there exists a set of
potential players $B$, with {\em ordered-like} values, one at each
vertex of the hypercube. Then, by using the results of Section
\ref{generalLB}, there exists a path where the vertices are {\em
  zig-zag}-ordered.
 
The separation into these two extreme cases was first used in
\cite{CRV10}. The crucial difference, is that for their problem the
protocol is specified independently of the underlying graph, and
therefore the adversary knows the case distinction (ordered or
shapley) and bases the lower bound construction on that. However, our
problem requires more work as the graph should be constructed {\em in
  advance}, and should work for both cases.

\begin{theorem}
  There exist graph metrics, such that the PoA of
  any universal cost-sharing protocol is at least $\Omega(\log k)$, where
  $k$ is the number of activated vertices.
\end{theorem}

\begin{proof}
Let $k=2^{r-1}+1$ be the number of activated vertices with $r \geq 4$, (hence $k \geq 9$). 

\noindent\textbf{Graph Construction.} 
We use as a base of our lower bound construction, a hypercube
$Q:=Q_{n}$, with edge costs equal to $1$ and $n=n(r)$ as in
Theorem~\ref{thm:badPath}. Based on $Q$, for $M =16k^{12} 2^{3n}$ we
construct the following network with $N = 2^{n} M $ vertices, plus the
designated root $t$. We add to $Q$ direct edges/shortcuts as follows:
for every two vertices $v,u$ of distance $2^j$, for $j\in [r]$, we add
an edge/shortcut, $(v,u)$, with cost equal to $\hat{c}_j =
2^j\left(\frac{k-1}{k}\right)^{j} = \Omega(2^j)$.  Moreover, for every
vertex $v_q$ of $Q$, we create $M-1$ new vertices, each of which we
connect with $v_q$ via a zero-cost edge. Let $V_q$ be the set of these
vertices (including $v_q$). Finally, we add a root $t$, which we
connect with every vertex $v_q$ of $Q$, via two edges $e_{q1}$ and
$e_{q2}$, with costs $2k$ and $2k\cdot k/6$, respectively. We denote
this new network by $Q^*$ (see Figure \ref{fig:allProtConstr}).

\begin{figure}[h]
\centering
\includegraphics[scale=0.35]{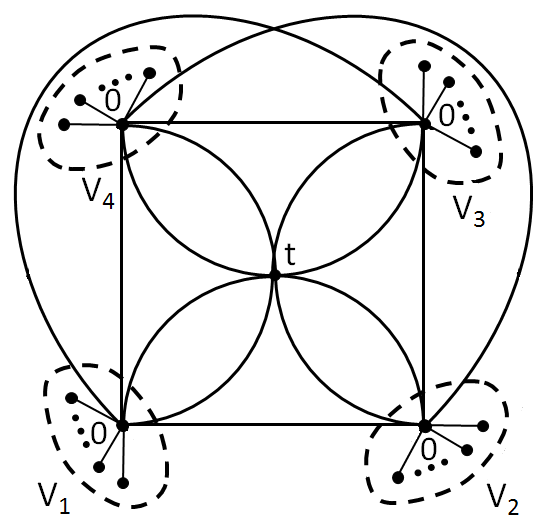} 
\caption{An example of $Q^*$ for $Q_2$ as the base hypercube.}
     \label{fig:allProtConstr}
\end{figure}

We will show that any GWSP for $Q^*$ 
has PoA $\Omega (\log k)$. Any GWSP can be described by concatenations
of potential-based cost-sharing protocols $\Xi_1, \ldots, \Xi_h$ for a
partition of the $V(Q^*)$ into $h$ subsets $U_1, \ldots , U_h$, where
$\Xi_j$ is induced by some edge potential $f_j$. Following the
analysis of Chen, Roughgarden and Valiant~\cite{CRV10}, we scale the
$f_j$'s such that for every $i,j$, $f_j(i) \geq 1$. For nonnegative
integers $s$ and for $\alpha = \left(1+k^{-3}\right)^{\frac{1}{2k}}$,
we form subgroups of vertices $A_{js}$, for each $U_j$, as
$A_{js}=\{i\in U_j: f_j(i) \in \left[\alpha^s, \alpha^{s+1}\right]\}$
(note that some of $A_{js}$'s may be empty).

The adversary proceeds in two cases, depending on the
intersection of the $A_{js}$'s with the $V_q$'s.

\noindent\textbf{Shapley-like cost-sharing.} Suppose first that there exist $A_{js}$ and $V_q$
such that $|A_{js} \cap V_q|\geq k$, and take a subset $R\subseteq
A_{js} \cap V_q$ with exactly $k$ vertices. The adversary will request
precisely the set $R$. Budget-balance implies that there exists some
vertex $i^*\in R$ which is charged at most $1/k$ proportion of the
cost. Moreover, Lemma~\ref{lem:proximity} implies that, all $i \in R$
are charged at most $\alpha (1/k+2k^{-2}) \leq 2\cdot (3/k) = 6/k$
proportion of the cost.

Note that there is a NE where all players follow the edge $e_{q2}$, with
cost $2k \cdot k/6$; no player's share is more than $2k$ and any
alternative path would cost at least $2k$. However, the optimum
solution is to use the parallel link $e_{q1}$ of cost $2k$. Therefore,
the PoA is $\Omega(k)$ for this case.

\noindent\textbf{Ordered-like cost-sharing.} If there is no such $R$
with at least $k$ vertices, then $|A_{js}\cap V_q|\leq k$ for all
$j,s$ and $q$, which means that each $A_{js}$ has size of at most $k
2^{n}$.
For every $j \in [h]$, we group consecutive sets $A_{js}$ (starting
from $A_{j0}$) into sets $B_{jl}$, such that each $B_{jl}$, (except
perhaps from the last one), contains {\em exactly} $4k^5$ {\em
  nonempty} $A_{js}$'s. The last $B_{jl}$ contains at most $4k^5$ {\em
  nonempty} $A_{js}$ sets. Consider the lexicographic order among
$B_{jl}$'s, i.e. $B_{jl} < B_{j'l'}$ if either $j < j'$ or $j = j'$
and $l<l'$. Rename these sets based on their total order as
$B_i$'s. The size of each $B_i$ is at most $4k^6 2^{n}$.

Now we apply Lemma \ref{lem:matching-s} on the set $N$, for $r =4k^6
2^{n}$ and $s=m=2^{n}$, by considering the subsets $V_q$ as the
partition of $N$ (recall that $|V_q| = M = r^2 s$).  As a coloring
scheme, we color all the vertices of each $B_i$ with the same color
and use different colors among the sets $B_i$.  Lemma
\ref{lem:matching-s} guarantees that for each $V_q$ there exists $V'_q
\subset V_q$ of size $2^{n}$, such that every $v\in V' = \cup_qV'_q$
belongs to a distinct $B_i$.

The order of $B_i$'s suggests an order of the vertices of $V'$.  Since
the $V'_q$'s form a partition of $V'$, Lemma~\ref{lem:nonconsequent}
guarantees the existence of a subset $C \subset V'$, such that $C$
contains exactly one vertex from each $V'_q$ and there are no
consecutive vertices in $C$. This means that $C$ contains exactly one
vertex from each set $V_q$ and all these vertices belong to {\em
  different} and {\em non-consecutive} sets $B_{i}$.

To summarize, so far we know that:
\begin{enumerate}[label=(\roman*)]
\item for any pair of vertices $v,u \in C$, either $v$ and $u$ come
  from different $U_j$'s or their $f_j(v)$ and $f_j(u)$ values differ
  by a factor of at least $\alpha^{4k^5} \geq 8k+1$ (since there exist
  at least $4k^5$ nonempty sets $A_{js}$ between the ones that $v$ and
  $u$ belong to).
\item $C$ is a copy of $Q_{n}$ (by ignoring zero-cost edges).
\end{enumerate} 

Let $\pi$ be the order of vertices of $C$ (recall that they are ordered according to the $B_i$'s they belong to). 
Theorem \ref{thm:badPath} guarantees that there always exists at least one
distance preserving path $P_{r}(\pi)$ (see Definition \ref{def:badOrder}).  Let $S$ be the
vertices of $P_{r}(\pi)$ excluding the last class $D_{r}$ (see
Definition~\ref{def:badOrderSets}). 
The adversary will activate this set $S$ ($|S| = k$). It remains to
show that there exists a NE, the cost of which is a factor of
$\Omega(\log k)$ away from optimum.
We will refer to these vertices as $S = \{s_1, s_2, \ldots, s_k\}$,
based on their order $\pi$, from smaller label to larger, and let player $i$ be associated with 
$s_i$.

Let $\mathcal{P'}$ be the class of strategy profiles $\P=(P_1,\ldots,
P_k)$ which are defined as follows:
\begin{itemize} 
\item $P_1=e_{11}$ and $P_2 = (s_1,s_2) \cup P_1$, 
where $(s_1,s_2)$ is the shortcut edge between $s_1$ and $s_2$. 
\item From $i=3$ to $k$, let $s_{\ell} \in \Pi(s_i)$ 
be one of $s_i$'s parents in the class hierarchy 
(we refer the reader to the beginning of Section~\ref{generalLB}); 
then $P_i = (s_i,s_{\ell}) \cup P_{\ell}$, 
where $(s_i,s_{\ell})$ is the shortcut edge between $s_i$ and $s_{\ell}$.
\end{itemize}
We show in Claim \ref{cl:PNEbestProtocol} 
that there exists a strategy profile $\P^* \in \mathcal{P'}$ which is a NE. 
$\P^*$ has cost: 
$$c(\P^*)={c}({e_{11}}) + \hat{c}_r+\sum_{j=1}^{r-1} |D_j|\cdot \hat{c}_{r-j} = 
\Omega(2^{r}) + \Omega(2^{r})+\sum_{j=1}^{r-1} 2^{j-1}\cdot \Omega(2^{r-j}) = 
\Omega(r2^r).$$ 

However, there exists the solution 
$P_{r}(\pi) \cup e_{11}$, which has cost of $O(2^r)$. 
Therefore, the PoA is $\Omega(r) = \Omega(\log k)$. 

\begin{claim}
\label{cl:PNEbestProtocol}
There exists $\P^* \in \mathcal{P'}$ which is a Nash equilibrium.
\end{claim}

\begin{proof} 
  We prove the claim by using better-response dynamics. Note that any
  GWSP induces a potential game for which better-response dynamics
  always converge to a NE (see \cite{CRV10, GMW13}). We start
  with some $\P_1 \in \mathcal{P'}$ and we prove that, after a sequence
  of players' {\em best}-responses, we end up in $\P_2 \in
  \mathcal{P'}$. Proceeding in a similar way we eventually converge to
   $\P^*$, which is the required NE.

   We next argue that for any $\P \in \mathcal{P'}$, players $1$ and
   $2$, have no incentive to deviate from $P_1$ (argument (a)) and
   $P_2$ (arguments (b)), respectively.  We further show that, given
   any strategy profile $\hat{\P}$, there exists some $\P \in
   \mathcal{G}$ such that: for every player $i \notin \{1,2\}$, if
   $\P^i = (P_1,\ldots,P_{i-1},\hat{P}_{i+1},\ldots,\hat{P}_k)$ are
   the strategies of the other players, $i$ prefers
   $P_i$ to $\hat{P}_{i}$ (arguments (c)-(e)).
We define the desired $\P$ recursively as follows: $P_1=e_{11}$, $P_2 = (s_1,s_2) \cup P_1$ and  
from $i=3$ to $k$, $P_i \in A=\arg \min_{P'_i} \{c_{i}(\P^i,P'_i)|\exists (P'_{i+1},\ldots,P'_k) 
\mbox{ s.t. } (P_1,\ldots,P_{i-1},P'_i,\ldots,P'_k)\in \mathcal{G}\}$.  
If $\hat{P}_{i} \in A$ then we set $P_i = \hat{P}_{i}$, otherwise we choose a path from $A$ arbitrarily.
 
We first give some bounds on players' shares.

\begin{enumerate}

\item Let $R \subseteq S$ be any set of players that use some edge $e$
  of cost $c_e$ and let $i$ be the one with the smallest
  label. 
  The total share of players $R\setminus\{i\}$ is upper bounded by
  $\sum_{i=1}^{|R|-1} \frac 1{(8k+1)^i+1} \cdot c_e < \frac {c_e}{8k}$
  (Lemma~\ref{lem:separation}).  Moreover, $i$'s share is at least
  $\frac{8k-1}{8k} c_e$.

\item
The total cost of any $P_i$ under $\P^i$, is at most $8k$. 
This is true because, for every player ${i'}$ with $i' \leq i$, 
the first edge of $P_{i'}$ is a shortcut to reach one of $s_{i'}$'s parents, with cost at most 
$2^{r-j}$, where $D_j$ is the class that $s_{i'}$ belongs to. Therefore, 
the cost of $P_i$ 
is at most $2k + \sum_{l=0}^{r-1} 2^{r-l} < 8k$. 
 
\item
By combining the above two arguments, under $\P^i$, the total share of player $i$ for the edges of $P_i$ at which  
she is not the first according to $\pi$, is at most $\frac 1{8k} \cdot 8k \leq 1$. 
\end{enumerate}

Here, we give the arguments for players $1$ and $2$. 
\begin{enumerate}[label=(\alph*)]
\item
The share of player $1$ under $\P$ is at most $2k$ and any other path
would incur a cost strictly greater than $2k$. 

\item
The share of player $2$ under $\P$ is at most
$2^r + 1 = 2k - 1$ (argument $3$), whereas if she doesn't connect through $s_1$, her
share would be at least $2k$. Moreover, if she connects to $t$ through $s_1$ but by using 
any other path rather than the shortcut $(s_1,s_2)$, the total cost of that path is 
at least $2^{r}\left(\frac{k-1}{k}\right)^{r-1}$. Player $2$ is first according to $\pi$ at that path  
and by argument $1$, her share is at least $2^{r}\frac{8k-1}{8k}\left(\frac{k-1}{k}\right)^{r-1} < \hat{c}_{r}$.

\end{enumerate}
We next give the required arguments in order to show that $P_i$ is a best response for 
player $i \neq \{1,2\}$ under $\P^i$. In the following, let $s_i \in D_j$ and let $s_{\ell}$ 
be the parent of $s_i$ such that $P_i = (s_i,s_\ell)\cup P_{\ell}$. 
Also let $s_{i'}$ be the predecessor of $s_i$, according to $\pi$, that is first met by following $\hat{P}_{i}$ 
from $s_i$ to $t$.  
\begin{enumerate}[resume,label=(\alph*)]

\item
Suppose that $s_{i'} = s_{\ell}$. 
\begin{itemize}
\item Assume that $\hat{P}_{i}$ doesn't use the shortcut $(s_i,s_\ell)$. The subpath 
of $\hat{P}_{i}$ from $s_i$ to $s_\ell$ contains edges at which $i$ is first
according to $\pi$ of total cost at least
$2^{r-j}\left(\frac{k-1}{k}\right)^{r-j-1}$. By argument $1$, her share is at least 
$2^{r-j}\frac{8k-1}{8k}\left(\frac{k-1}{k}\right)^{r-j-1} < \hat{c}_{r-j}$.
\item Assume that $\hat{P}_{i}$ doesn't use $P_{\ell}$. The subpath 
of $\hat{P}_{i}$ from $s_\ell$ to $t$ contains edges at which $i$ is first
according to $\pi$ of total cost at least $2$ (the minimum distance between two activated vertices). 
By argument $1$, her share is at least $2\frac{8k-1}{8k} > 1$, where $1$ is at most her share for 
$P_{\ell}$ (argument $3$). 
\end{itemize}  

\item 
Suppose that $s_{i'}$ is $s_i$'s other parent. 
If $\hat{P}_{i} \neq (s_i,s_{i'})\cup P_{i'}$, the above arguments 
still hold and so $c_{i}(\P^i,P_i) < c_{i}(\P^i,\hat{P}_{i})$. 
Otherwise, by the definition of $P_i$, either $P_i =  \hat{P}_{i}$, or 
$c_{i}(\P^i,P_i) < c_{i}(\P^i,\hat{P}_{i})$. 

\item 
Suppose that $s_{i'}$ is not a parent of $s_i$. 
Player $i$'s share in $P_i$ is at most $\hat{c}_{r-j}$
for her first edge/shortcut and at most $1$ for the rest of her path (argument $3$).
However, all edges that are used by players that
precedes $i$ in $\pi$ have cost at least $\hat{c}_{r-j}$.
Therefore, in $\hat{P}_{i}$, player $i$ is the first according to $\pi$ for edges of total cost at least
$\hat{c}_{r-j+1}$. This implies a cost-share of at least $\frac{8k-1}{8k}
\hat{c}_{r-j+1}$ (argument $1$). But for $k \geq 6$ and $j < r$, $\frac{8k-1}{8k}
\hat{c}_{r-j+1} > \hat{c}_{r-j} + 1$.

\end{enumerate}

We now describe a sequence of best-responses from some $\hat{\P} \in \mathcal{P'}$ to $\P$ 
($\P$ is constructed based on $\hat{\P}$ as described above).  
We follow the $\pi$ order of the players and for each player we apply her best response. 
First note that players $1$ and $2$ have no better response, so $P_1 = \hat{P}_1$ and $P_2 = \hat{P}_2$. 
When we process any other player $i$, we have already processed all her predecessors in $\pi$ and so, 
the strategies of the other players are $\P^i$. 
Therefore, $P_i$ is the best response for $i$ (it may be that $P_i=\hat{P}_i$, where no better response exists for $i$). 
The order that we process the vertices guarantees that $\P \in \mathcal{P'}$.

\end{proof}
\end{proof}


\section{Outerplanar Graphs}

In this section we show that there exists a class of graph metrics,
prior knowledge of which can dramatically improve the performance of
good network cost-sharing design. For {\em outerplanar} graphs, we
provide a universal cost-sharing protocol with constant PoA. In
contrast, we stress that uniform protocols cannot achieve PoA better
than $\Omega (\log k)$, because the lower bound for the greedy
algorithm of the OSTP can be embedded in an outerplanar graph (see
Figure \ref{fig:circle}(a) for an illustration).
	
We next define an ordered universal cost-sharing protocol $\Xi_{tour}$, and
we show that it has constant PoA. W.l.o.g. we assume that the metric
space is defined by a given {\em biconnected} outerplanar
graph\footnote{If it is not already biconnected, we turn it into an
  {\em equivalent} biconnected graph, by appropriately adding edges of
  infinity cost. By {\em equivalent} we mean that any NE outcome and
  the minimum Steiner tree solution remain unchanged after the
  transformation. Equivalence is obvious since we only add edges of
  infinity costs that cannot be used in neither any NE nor the minimum
  Steiner tree outcome.}. Every biconnected graph admits a {\em
  unique} Hamiltonian cycle~\cite{Sys79} that can be found in linear
time~\cite{CFN85}. $\Xi_{tour}$ orders the vertices according to the cyclic
order in which they appear in the Hamiltonian tour, starting from $t$
and proceeding in a clockwise order $\pi$.  In Figure
\ref{fig:circle}(a), $\pi(q_8) < \pi(q_4) < \pi(q_9) <
\ldots<\pi(q_{15})$.

{\footnotesize
\begin{figure}[h]
\centering
\begin{tabular}{c c}
\includegraphics[scale=0.3]{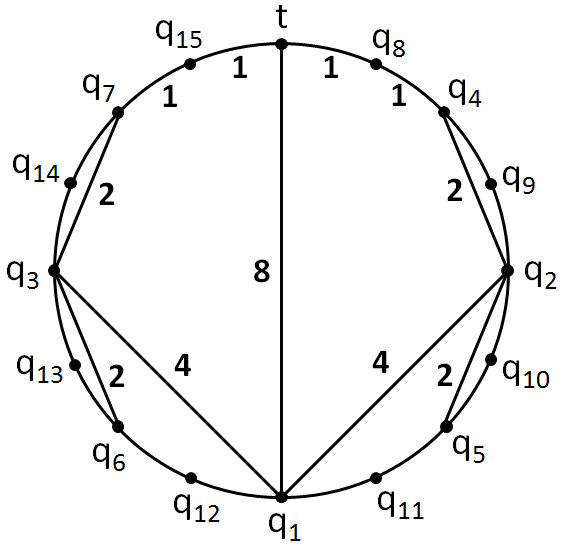} \qquad & \qquad 
\includegraphics[scale=0.3]{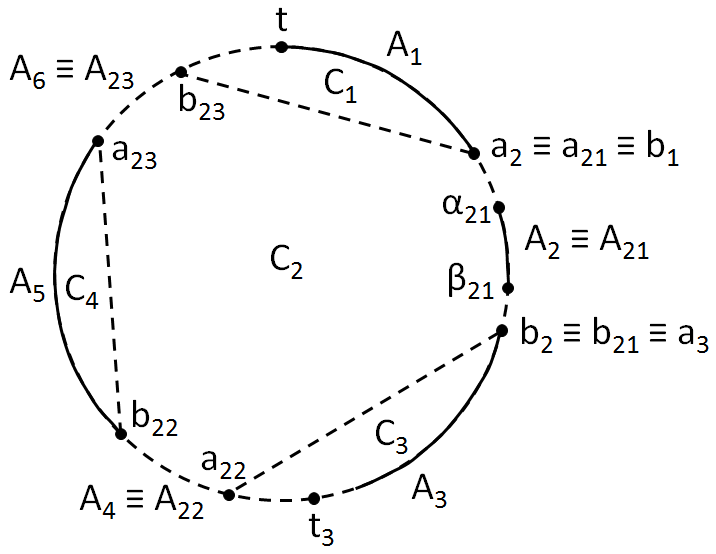} \\
$(a)$&$(b)$
\end{tabular}
\caption{\footnotesize $(a)$ shows an example of an outerplanar graph
  where the order $q_i < q_{i+1}$ gives PoA of $\Omega(\log k)$. $(b)$
  illustrates some elements from the proof of
  Theorem~\ref{thm:UBouter}, focusing on cycle $C_2$. The dashed
  components represent the optimum tree $T^*$.}
     \label{fig:circle}
\end{figure}
}

As a warm-up, we first bound from above the PoA of $\Xi_{tour}$ for cycle
graphs, and then extend it to all outerplanar graphs.

\begin{lemma}
\label{lem:cycle}
The PoA of $\Xi_{tour}$ in cycle graphs is at most $2$.
\end{lemma}

\begin{proof}
  Consider a cycle graph $C = (V,E,t)$ and let $S\subseteq V$ be the
  set of the activated vertices.  Let $T^*$ be the minimum Steiner
  tree (path) that connects $S\cup\{t\}$, and $a$, $b$ be its two
  endpoints. Note that minimality of $T^*$ implies that $a, b \in S
  \cup \{t\}$. $a$ and $b$ partition $C$ into two paths
  $(T^*,C\setminus T^*)$
  and $t$ divides further $T^*$ into two paths $P'_1$,
  $P'_2$. 
Let $S_1 = \{u_1, \ldots, u_{r}=a\}$ and $S_2=\{w_1, \ldots, w_{s}=b\}$ be
the activated vertices of $P'_1$ and $P'_2$, respectively. W.l.o.g., assume that 
$\pi(u_i) < \pi(u_{i+1})$ and $\pi(w_{j+1}) <\pi(w_j)$, for all $i,j$.

Consider any NE, $\P=(P_i)_{i\in N}$. We bound from above the share of
each player $v \neq w_{s}$, by its distance from their immediate
predecessor in $\pi$, as follows.  By adopting the convention that
$u_0=t$,
\begin{eqnarray*}
c_{u_i}(\P) \leq d(u_i,u_{i-1}), \;\;\forall i \in [r],
\qquad\qquad c_{w_j}(\P) \leq d(w_j,w_{j+1}), \;\;\forall j \in [s - 1].
\end{eqnarray*} 
Also $c_{w_{s}}(\P) \leq d(w_{s},t)$. Overall,
\begin{eqnarray*}
c(\P) &=& \sum_{v \in S} c_v(\P) \leq 
    \sum_{u_i \in S_1} d(u_i,u_{i-1}) + \sum_{w_j \in S_2 - \{w_{s}\}} d(w_j,w_{j+1})+ d(w_{s},t)\\
	& \leq& c(P'_1) + c(P'_2) + c(P'_2) \leq 2c(T^*).
\end{eqnarray*}
\end{proof}

\begin{theorem}
\label{thm:UBouter}
The PoA of $\Xi_{tour}$ in outerplanar graphs is at most $8$.  
\end{theorem}

\begin{proof}
  Based on the previous discussion, it is sufficient to consider only
  biconnected outerplanar graphs with non-negative costs, including
  infinity. Let $G=(V,E,t)$ be any such graph with $S$ being the set
  of activated vertices.

  Let $T^*$ be the minimum Steiner tree that connects $S \cup \{t\}$,
  and $C$ be the unique Hamiltonian tour of $G$, forming its outer
  face. Let $E^*=E(T^*)\setminus E(C)$ be the set of non-crossing
  chords of $C$ that belong to $T^*$.
  Then $C\cup E^*$ forms $|E^*| +1 = r$ cycles $C_1, \ldots, C_r$,
  where every pair $C_i$, $C_j$ are either edge-disjoint or they have
  a single common edge belonging to $E^*$. On the other hand, each
  edge of $C$ belongs to exactly one $C_i$ and each edge of $E^*$
  belongs to exactly two $C_i$'s. Figure \ref{fig:circle}(b) provides
  an illustration.

  For every $i\in [r]$, let $S_i = (S \cup \{t\}) \cap V(C_i)$ be the
  activated vertices that lie in $C_i$ and $t_i$ be the vertex that is
  first in $\pi$ among $S_i$.  W.l.o.g. assume that, for all $i \in
  [r-1]$, $\pi(t_i) \leq \pi(t_{i+1})$ (then $t_1=t$).  Also let
  $T^*_i$ be the subgraph of $T^*$ that intersects with $C_i$. Then
  $T^*_i$ should be a path connecting $S_i$.

  Consider any NE, $\P=(P_i)_{i\in S}$. We show separately that the
  shares of all $S_i \setminus \{t_i\}$ are bounded by $4 c(T^*)$ and
  the shares of all $t_i$'s are bounded by $4 c(T^*)$.
 
For the first case we use Lemma~\ref{lem:cycle}. For any cycle $C_i$,
by considering $t_i$ as the root, Lemma~\ref{lem:cycle}
provides a bound on the shares of $S_i \setminus \{t_i\}$. So, $\sum_{v\in S_i \setminus 
  \{t_i\}} c_v(\P) \leq 2c(T^*_i)$.  Recall, that each edge of $E(T^*)$
belongs to at most two $C_i$'s, so by summing over all $i \in [r]$, 
$$ \sum_{i \in [r]}\sum_{v\in S_i \setminus \{t_i\}} c_v(\P) 
\leq 2\sum_{i \in [r]} c(T^*_i) \leq 4c(T^*).$$
 
The second case requires more careful treatment. The endpoints
of the edges of $E^*$ 
divide $C$ into a partition of nonzero-length arcs, $A_1, \ldots,
A_n$, named based on their clockwise appearance in $C$, starting from
an arc containing $t$.  For every $j \in [n]$, let $a_j$ and $b_j$ be
the two endpoints of $A_j$.  
The share of each $t_i$ can be bounded by its distance from $t_{i-1}$, for $i >1$ (recall 
that $t_1=t$). 
Let $A_{s_i}$ be an arc that $t_i$ lies, then 
\begin{eqnarray*}
\sum_{i\in [2,r]} c_{t_i}(\P) \leq \sum_{i\in [2,r]}d(t_i, t_{i-1}) \leq \sum_{i\in [r]}( d(a_{s_i}, t_i) + d(t_{i},b_{s_i})) + 
\sum_{j \in [n]\setminus\{s_1,\ldots,s_r\}} d(a_{j},b_{j})=F.
\end{eqnarray*}

We next upper bound $F$ by $\sum_{i\in [r]}2c(T_i^*)$. 
Note that each arc $A_j$ belongs to exactly one $C_i$ and every $C_i$ contains at least one such arc 
(otherwise $T^*$ would have a cycle). 
We concentrate to a specific $C_i$ and show that the portion of $F$ associated with 
$C_i$'s arcs is upper bounded by $2c(T_i^*)$. 

Let $A_{i1},...,A_{i{n_i}}$ be the arcs belonging 
to $C_i$ and $a_{ij}$, $b_{ij}$ be the endpoints of $A_{ij}$. 
Also let $A_{is}$ be an arc containing $t_i$. 
Recall that $T^*_i$ is a path and every edge of $E(C_i)\cap E^*$ belongs to 
$T^*_i$. Therefore, $T^*_i$ contains entirely all but one $A_{ij}$, say $A_{im}$ 
(see also Figure~\ref{fig:circle}(b)). 
We examine the two cases of $m=s$ and $m\neq s$ separately. 

\noindent\textbf{Case 1: }$\mathbf{m=s}.$ $a_{is}$, $b_{is}$ (as endpoints of edges of $E^*$) 
and $t_i$ are vertices of the path $T_i^*$. Therefore, either some path from $t_i$ to $a_{is}$ 
or some path from $t_i$ to $b_{is}$ belongs to $T_i^*$; w.l.o.g. assume that it is some path from $t_i$ to $a_{is}$. 
Then $\sum_{j\in [n_i],j\neq s} d(a_{ij},b_{ij}) + d(t_i,a_{is}) \leq c(T^*_i)$. Moreover, 
since $b_{is}$ and $t_i$ are vertices of $T_i^*$, $d(t_i,b_{is}) \leq c(T^*_i)$. 

\noindent\textbf{Case 2: }$\mathbf{m\neq s}.$
Similarly, $\sum_{j\in [n_i],j\neq m,s} d(a_{ij},b_{ij}) + d(t_i,a_{is}) + d(t_i,b_{is}) \leq c(T^*_i)$. 
Also $a_{im}$ and $b_{im}$ are vertices of $T^*_i$ and hence, $d(a_{im},b_{im})\leq c(T^*_i)$. 

To sum up, in both cases it holds that  
$$\sum_{j\in [n_i],j\neq s} d(a_{ij},b_{ij}) + d(t_i,a_{is}) + d(t_i,b_{is}) \leq 2c(T^*_i).$$
By summing over all $i$, 
$F \leq \sum_{i\in [r]}2c(T_i^*) \leq 4 c(T^*).$ 
Finally, by summing over the whole $S$, $c(\P) = \sum_{v\in S} c_v(\P) \leq 8c(T^*)$.
\end{proof}


\section{Stochastic Network Design}
\label{sec:stochastic}

In this section we study the {\em stochastic} model, where the set of active vertices is drawn 
from some probability distribution $\Pi$.   
Each vertex $v$ is activated independently  
with probability $p_v$;  
the set of the activated vertices are no longer 
picked adversarially, but it is sampled based on the probabilities $p_v$'s, i.e., 
the probability that set $S$ is active is $\Pi(S) = \prod_{v\in S} p_v \cdot \prod_{v \notin S} (1-p_v)$. 
On the other hand, the probabilities $p_v$'s (and therefore $\Pi$), are chosen adversarially.  
The cost sharing protocol is decided by the designer without the knowledge of the activated set 
and the designer may have knowledge of $\Pi$ or access to some oracle of $\Pi$.

We show that there exists a {\em randomized} ordered protocol that achieves 
constant PoA. This result holds even for the 
{\em black-box} model \cite{ST08}, meaning that 
the probabilities are not known to the designer, however she is allowed to draw 
independent (polynomially many) samples. 
On the other hand, if we assume that the probabilities $p_v$'s are known to the designer, 
there exists a {\em deterministic} ordered protocol that achieves 
constant PoA. We note that both protocols can be determined in polynomial time. 

The result for the randomized protocol depends on   
approximation ratios of the minimum Steiner tree problem. More precisely, given 
an $\alpha$-approximate minimum Steiner tree, we show an upper bound of $2(\alpha+2)$. 
The approximate tree is used in our algorithm as a base in order to construct a spanning tree, 
which finally determines an order of all vertices; the detailed algorithm is given in 
Algorithm \ref{alg:stochastic}. This algorithm and its slight variants have been used in 
different contexts: rend-or-buy problem \cite{GKPR07}, a priori TSP \cite{ST08} 
and, stochastic Steiner tree problem \cite{GGLS08}. 
  
\begin{algorithm}[H]
  \SetAlgorithmName{Algorithm}{algorithm}{List of Algorithms}
  \KwIn{A rooted graph $G=(V,E,t)$ and an oracle for the probability distribution $\Pi$.}
  \KwOut{$\Xi_{{rand}}$.}
  \begin{itemize}
	\item Choose a random set of vertices $R$ by drawing from distribution $\Pi$ 
	and construct an $\alpha$-approximate minimum Steiner tree, 
	$T_{\alpha}(R)$, over $R\cup\{t\}$.  
	\item Connect all other vertices 
	$V \setminus V(T_{\alpha}(R))$ with their nearest neighbor in $V(T_{\alpha}(R))$ (by breaking ties arbitrarily). 
	\item Double the edges of that tree and traverse some Eulerian tour starting from $t$. Order 
	the vertices based on their first appearance in the tour. 
	\end{itemize}
  \caption{Randomized order protocol $\Xi_{{rand}}$}
  \label{alg:stochastic}
\end{algorithm}

\begin{theorem}
\label{thm:stochasticRandomized}
Given an $\alpha$-approximate solution of the minimum Steiner tree
problem, $\Xi_{{rand}}$ has PoA at most $2(\alpha+2)$.
\end{theorem}

\begin{proof}
Let $\pi$ be the order of all vertices $V$, defined by $\Xi_{{rand}}$, and $S$ be the random set of activated vertices 
that require connectivity with $t$. For the rest of the proof we denote by $MST(S)$ 
a minimum spanning tree over $S\cup\{t\}$.

	Let $s_1, \ldots, s_r$ be the vertices of $S$ as appeared in $\pi$ and the strategy profile $\P_R(S)=(P_1,\ldots,P_r)$ 
	be a NE of set $S$. Under the convention that $s_0 = t$, $c_{s_i}(\P_R(S)) \leq d_G(s_i,s_{i-1})$ for all $s_i \in S$. 
	We construct a tree $T_{R,S}$ 
from the $T_{\alpha}(R)$ of Algorithm \ref{alg:stochastic}, by connecting only all vertices of 
	$S \setminus V(T_{\alpha}(R))$ with their nearest neighbor in $V(T_{\alpha}(R))$ 
	(by breaking ties in accordance to Algorithm \ref{alg:stochastic}). Note that, by doubling the edges of $T_{R,S}$, 
	there exists an Eulerian tour starting from $t$, where the order of   
	the vertices $S$ (based on their first appearance in the tour) is $\pi$  
	restricted to the set $S$. Therefore, 
	$\sum_{s_i \in S} d_{T_{R,S}}(s_i,s_{i-1}) +d_{T_{R,S}}(s_0,s_r)= 2 c(T_{R,S})$. By combining the above, 
	
	\begin{eqnarray}
	c(\P_R(S)) = \sum_{s_i\in S} c_{s_i}(\P_R(S)) \leq \sum_{s_i\in S} d_G(s_i,s_{i-1}) \leq  \sum_{s_i\in S} d_{T_{R,S}}(s_i,s_{i-1}) \leq 2 c(T_{R,S}). \label{boundStrategies}
	\end{eqnarray}
	Let $D_v(R)$ be the distance of $v$ from its nearest neighbor in $(R\cup\{t\})\setminus\{v\}$. 
	In the special case that $v=t$, we define $D_v(R)=0$ 
	Then, 
	\begin{eqnarray}
	c(T_{R,S}) = c(T_{\alpha}(R))+\sum_{v\in S\setminus V(T_{\alpha}(R))} D_v(V(T_{\alpha}(R))) \leq c(T_{\alpha}(R))+\sum_{v\in S} D_v(R). \label{T_RA}
	\end{eqnarray} 
	We use an indicator $I(v\in S)$ which is $1$ when $v \in S$ and $0$ otherwise; then 
	$\sum_{v\in S} D_v(R) = \sum_{v}I(v\in S)D_v(R)$. By taking the expectation over $R$ and $S$, 
	
	$$ \E_R[\E_S[c(T_{R,S})]] \leq  \E_R\left[c(T_{\alpha}(R))\right]+\E_R\left[\E_S\left[\sum_{v \in V}I(v\in S)D_v(R)\right]\right].$$
	Since $S$ and $R$ are independent samples we can bound the second term as:
	\begin{eqnarray}
	\E_R\left[\E_S\left[\sum_{v \in V}I(v\in S)D_v(R)\right]\right] 
	&=& \sum_{v \in V}\E_S[I(v\in S)]\E_R[D_v(R)] 
	= \sum_{v \in V}\E_S[I(v\in S)]\E_S[D_v(S)] \notag\\
	&=& \E_S\left[\sum_{v \in V}I(v\in S)D_v(S)\right]
	\leq \E_S[c(MST(S))]. \label{indExp}
	\end{eqnarray}
	The last equality holds since $D_v(S)$ is the distance of $v$ from its nearest neighbor in $(S\cup\{t\})\setminus\{v\}$ 
	and it is independent of the event $I(v\in S)$. For the inequality, $D_v(S)$ is upper bounded by 
	the minimum distance of $v$ from its parent in the $MST(S)$. 
	Let $T^*_S$ be the minimum Steiner tree over $S\cup\{t\}$, 
	then it is well known that $c(MST(S))\leq 2c(T^*_S)$. Overall, 
	$$ \E_R[\E_S[c(\P_R(S))]] \leq 2\E_R[\E_S[c(T_{R,S})]]\leq 2 (\E_S[c(T_{\alpha}(S))] + \E_S[c(MST(S))]) \leq 2(\alpha +2) \E_S[c(T^*_S)].$$  
	
\end{proof}

By applying the $1.39$-approximation algorithm of \cite{BGRS10} we get the following.

\begin{corollary}
$\Xi_{{rand}}$ has PoA at most $6.78$.
\end{corollary}

\begin{theorem}
\label{thm:stochasticDeterministic}
There exists a deterministic ordered protocol 
with PoA at most $16$.
\end{theorem}

\begin{proof}
We use derandomization techniques similar to \cite{WZ07,ST08} and for completeness we give the full proof here. 
First we discuss how we can get a PoA of $6.78$, if we drop the requirement of determining the protocol in polynomial time. 
	Similar to the proof of Theorem~\ref{thm:stochasticRandomized} we define the tree $T_{R,S}$ for the random activated set $S$ 
	as follows: we construct $T_{R,S}$ 
from the $T_{\alpha}(S)$ of Algorithm \ref{alg:stochastic}, by connecting only all vertices of 
	$S \setminus V(T_{\alpha}(R))$ with their nearest neighbor in $V(T_{\alpha}(R))$ 
	(by breaking ties in accordance to Algorithm \ref{alg:stochastic}). 
We apply the standard derandomization approach of {\em conditional expectation method} on $T_{R,S}$. 
More precisely, we construct a deterministic set $Q_1$ to replace the random $R$ in Algorithm~\ref{alg:stochastic}, 
by deciding for each vertex of $V\setminus \{t\}$, one by one, whether it belongs 
to $R$ or not. Assume that we have already processed the set $Q \subset V$ and we have decided that 
for its partition $(Q_1,Q_2)$, $Q_1 \subseteq R$ and $Q_2 \cap R=\emptyset$ (starting from $Q_1=\{t\}$ and $Q_2=\emptyset$). Let $v$ be the next vertex to be processed. 
From the conditional expectations and the independent activations we know that 
\begin{eqnarray*}
\E_{S,R}[c(T_{R,S})| Q_1 \subseteq R, Q_2 \cap R=\emptyset]
&=& \E_{S,R}[c(T_{R,S})| Q_1 \subseteq R, Q_2 \cap R=\emptyset,v \in R]p_v\\
&&+\E_{S,R}[c(T_{R,S})| Q_1 \subseteq R, Q_2 \cap R=\emptyset,v\notin R](1-p_v),
\end{eqnarray*}
meaning that  
\begin{eqnarray*}
&\mbox{either} &\E_{S,R}[c(T_{R,S})| Q_1 \subseteq R, Q_2 \cap R=\emptyset,v \in R] \leq \E_{S,R}[c(T_{R,S})| Q_1 \subseteq R, Q_2 \cap R=\emptyset],\\ 
&\mbox{or}&\E_{S,R}[c(T_{R,S})| Q_1 \subseteq R, Q_2 \cap R=\emptyset,v\notin R] \leq \E_{S,R}[c(T_{R,S})| Q_1 \subseteq R, Q_2 \cap R=\emptyset].
\end{eqnarray*} 
In the first case we add $v$ in $Q_1$ and in the second case we add $v$ in $Q_2$.  
Therefore, after processing all vertices, $E_S[c(T_{Q_1,S})] \leq \E_{S,R}[c(T_{R,S})]$. If we replace the sampled $R$ of Algorithm \ref{alg:stochastic} 
with the deterministic set $Q_1$, we can get the same bound on the PoA with the randomized protocol of Theorem \ref{thm:stochasticRandomized}.

However, the value of $\E_{S,R}[c(T_{R,S})| Q_1 \subseteq R, Q_2 \cap R=\emptyset]$ seems difficult to be computed in polynomial time; 
the reason is that it involves the computation of $\E_{R}[c(T_{\alpha}(R))| Q_1 \subseteq R, Q_2 \cap R=\emptyset]$ 
which seems hard to be handled.  
To overcome this problem we use an estimator $EST(Q_1,Q_2)$ of $\E_{S,R}[c(T_{R,S})| Q_1 \subseteq R, Q_2 \cap R=\emptyset]$, 
which is constant away from the optimum $\E_S[c(T^*_S)]$, 
where $T^*_S$ is the minimum Steiner 
tree over $S\cup\{t\}$. 
Following \cite{WZ07,ST08}, we use the optimum solution of the relaxed Connected Facility Location Problem (CFLP) on $G$ in order 
to construct a feasible solution $\bar{\y}$ of the relaxed Steiner Tree Problem (STP) for a given set $R$. We show that 
the objective's value of the fractional STP for $\bar{\y}$ is constant away from $\E_S[c(T^*_S)]$ and 
that its (conditional) expectation over $R$ can be efficiently computed. 
This quantity is used in order to construct the estimator $EST(Q_1,Q_2)$. 
We apply the method of conditional expectations on $EST(Q_1,Q_2)$ and after processing all vertices, by using 
the primal-dual algorithm \cite{GW95}, we compute a Steiner tree on $Q_1$ with cost no more than twice the cost of the fractional solution. 

In the rooted CFLP, a rooted graph $G=(V,E,t)$ is given and the designer should select some facilities to open,  
including $t$, and connects them via some Steiner tree $T$. Every other vertex is assigned to some facility. 
The cost of the solution is $M$ ($M>1)$ times the cost of $T$, plus the distance of every other vertex 
from its assigned facility. Our analysis requires to consider a slightly different cost of the solution, 
which is the cost of $T$, plus the distance of every other vertex $v$ from its assigned facility multiplied by $p_v$. 
In the following LP relaxation of the CFLP, $z_{e}$ and $x_{ij}$ are $0$-$1$ variables indicate, respectively, if 
$e \in E(T)$ and whether the vertex $j$ is assigned to facility $i$. 
$\delta(U)$ denotes the set of edges with one endpoint in $U$ and the other in $V\setminus U$, 
$d(i,j)$ denotes the minimum distance between vertices $i$ and $j$ in $G$ and $c_e$ is the cost of edge $e$.   

\vspace{10pt}
\begin{tabular}{|l r  l r|}
\hline
\multicolumn{4}{|l|}{LP1: CFLP}
\\\hline
&$\min$ &$B+C$ &\\
subject to & $\sum_{i \in V} x_{ij}$ =& 1& $\forall j \in V $\\
&$\sum_{e \in \delta(U)}z_e\geq$&$\sum_{i\in U}x_{ij}$&$\forall j\in V, \forall U \subseteq V\setminus\{t\}$\\
&$B=$&$\sum_{e \in E}c_e z_{e}$&\\
&$C=$&$\sum_{j\in V} p_j\sum_{i \in V}d{(i,j)}x_{ij}$&$$\\
&$z_{e},x_{ij}\geq$&$0$&$\forall i,j \in V$ and $\forall e\in E$\\\hline
\end{tabular}
\vspace{10pt}

Let $(\z^*=(z^*_{e})_{e},\x^*=(x^*_{ij})_{ij},B^*,C^*)$ be the optimum solution of LP1. 
\begin{claim}
\label{B+C<3Opt}
$B^*+C^* \leq 3\E_S[c(T_S^*)]$. 
\end{claim}
\begin{proof}
Given a set $S \subseteq V$, for every edge $e \in T^*_S$ ($T^*_S$ is the minimum Steiner tree over 
$S\cup\{t\}$) let $z_e = 1$ and for $e \notin T^*_S$ let $z_e = 0$. 
For every $j \in V$ let $x_{ij} = 1$ if $i$ is $j$'s  
nearest neighbor in $(S\cup\{t\})\setminus \{j\}$. 
Set the rest of $x_{ij}$ equal to $0$. Note that this is a feasible solution of LP1 
with objective value $B_S+C_S \leq c(T^*_S)+\sum_{v\in V} p_v D_v(S)$. By taking the expectation 
over $S$, 
\begin{eqnarray*}
B^*+C^* &\leq& \E_S[B_S+C_S] \leq \E_S[c(T^*_S)]+\sum_{v\in V} \E_S[I(v\in S)] \E_S[D_v(S)] = \E_S[c(T^*_S)]+ \E_S\left[\sum_{v\in S} D_v(S)\right]\\
&\leq& \E_S[c(T^*_S)]+\E_S[c(MST(S))]\leq 3\E_S[c(T^*_S)],
\end{eqnarray*}
\end{proof}

By using the solution $(\z^*=(z^*_{e})_{e},\x^*=(x^*_{ij})_{ij},B^*,C^*)$, we construct a feasible solution for the 
following LP relaxation of the STP over some set $R\cup\{t\}$.  

\vspace{10pt}
\begin{tabular}{|l r  l r|}
\hline
\multicolumn{4}{|l|}{LP2: STP over $R\cup \{t\}$}
\\\hline
&$\min$ &$\sum_{e \in E}c_{e}y_{e}$ &\\
subject to & $\sum_{e \in \delta(U)}y_{e}\geq$&$1$& $\forall U \subseteq V\setminus\{t\}:R\cap U \neq \emptyset$\\
&$y_{e}\geq$&$0$&$\forall e \in E$\\\hline
\end{tabular}
\vspace{10pt}

We define $a_{ij}(e)=1$ if $e$ lies in the shortest path between $i$ and $j$ and it is $0$ otherwise.  
For every edge $e$ we set $\bar{y}_{e}=z^*_{e}+\sum_{j\in R} \sum_{i \in V} a_{ij}(e)x^*_{ij}$.

\begin{claim}
$\bar{\y} = (\bar{y}_{e})_{e}$ is a feasible solution for LP2.
\end{claim}
\begin{proof}
The proof is identical with the one in \cite{WZ07} but we give it here for completeness. 
Consider any set $U\subseteq V\setminus\{t\}$ such that $R\cap U \neq \emptyset$ and let $\ell \in R\cap U$. It follows that 

\begin{eqnarray*}
\sum_{e \in \delta(U)}\bar{y}_{e} &\geq& \sum_{e \in \delta(U)}z^*_{e} + \sum_{e \in \delta(U)}\sum_{j \in R} \sum_{i \in V} a_{ij}(e) x^*_{ij} 
\geq \sum_{i \in U}x^*_{i\ell} + \sum_{e \in \delta(U)} \sum_{i \in V} a_{i\ell}(e) x^*_{i\ell} \\
&\geq& \sum_{i \in U}x^*_{i\ell} + \sum_{i \notin U} x^*_{i\ell}  \sum_{e \in \delta(U)}a_{i\ell}(e) 
\geq  \sum_{i \in U}x^*_{i\ell} + \sum_{i \notin U} x^*_{i\ell}  = 1.
\end{eqnarray*}
For the last inequality, note that $a_{i\ell}(e)$ should be $1$ for at least one $e \in \delta(U)$ since $i \notin U$ and $\ell \in U$.
\end{proof}

\begin{claim}
\label{est<B+C}
Let $\bar{c}_{ST}(R)$ be the cost of the objective of LP2 induced by the solution $\bar{\y}$. 
Then $\E_R[\bar{c}_{ST}(R)] = B^*+C^*$.
\end{claim}
\begin{proof}
\begin{eqnarray*}
\E_R\left[\bar{c}_{ST}(R)\right] &=& \E_R\left[\sum_{e\in E} c_e (z^*_{e}+\sum_{j\in R} \sum_{i \in V} a_{ij}(e)x^*_{ij})\right] 
 = B^* +\E_R\left[\sum_{j\in R} \sum_{i \in V} \sum_{e\in E} c_e a_{ij}(e)x^*_{ij}\right] \\
&=& B^* +\E_R\left[\sum_{j\in R} \sum_{i \in V} d(i,j)x^*_{ij}\right] 
= B^* +\sum_{j\in V} p_j \sum_{i \in V} d(i,j)x^*_{ij} 
= B^* + C^*.
\end{eqnarray*}
\end{proof}

Observe that due to the expression of $\bar{\y}$ we can efficiently 
compute any conditional expectation $\E[\bar{c}_{ST}(R)| Q_1 \subseteq R, Q_2 \cap R=\emptyset]$; 
this is because 
$$\E_R\left[\sum_{j\in R} \sum_{i \in V} a_{ij}(e)x^*_{ij}|Q_1 \subseteq R, Q_2 \cap R=\emptyset\right] = 
\sum_{j\in Q_1} \sum_{i \in V} a_{ij}(e)x^*_{ij} + \sum_{j\notin Q_1\cup Q_2} p_j\sum_{i \in V} a_{ij}(e)x^*_{ij}.$$
We further define $c_{C}(R) = \sum_{v\in V} p_v D_v(R)$. We can also efficiently compute any conditional 
expectation $\E[c_{C}(R)| Q_1 \subseteq R, Q_2 \cap R=\emptyset]$ (Claim 2.1 of \cite{WZ07}).  
We are ready to define our estimator: 
$$EST(Q_1,Q_2)=2\E_R[\bar{c}_{ST}(R)|Q_1 \subseteq R, Q_2 \cap R=\emptyset]+\E_R[\bar{c}_{C}(R)|Q_1 \subseteq R, Q_2 \cap R=\emptyset].$$
 
Our goal is to define a deterministic set $R^*$ to replace the sampled $R$ of Algorithm \ref{alg:stochastic}. 
We process the vertices one by one and we decide if they belong to $R^*$ by using the model conditional expectations on $EST(Q_1,Q_2)$. 
More specifically, assume that we have already processed the sets $Q_1$ and $Q_2$ (starting from $Q_1=\{t\}$ and $Q_2=\emptyset$)
 such that $Q_1 \subseteq R^*$ and 
$Q_2 \cap R^* = \emptyset$. Let $v$ be the next vertex to be processed. 
From the conditional expectations and the independent activations we know that 
$EST(Q_1,Q_2)=p_v EST(Q_1\cup\{v\},Q_2)+(1-p_v)EST(Q_1,Q_2\cup\{v\}).$ 
If $EST(Q_1\cup\{v\},Q_2) \leq EST(Q_1,Q_2)$ we add $v$ to $Q_1$, otherwise we add $v$ to $Q_2$. 
After processing all vertices and by using Claims \ref{B+C<3Opt} and \ref{est<B+C}, 
\begin{eqnarray*}
EST(R^*,V\setminus R^*) &\leq& EST(\{t\},\emptyset) \leq 6\E_S[c(T_S^*)] + \sum_{v\in V} p_v E_R[D_v(R)] \\
&=&6\E_S[c(T_S^*)] + E_R\left[\sum_{v\in V} I(v\in R) D_v(R)\right] \\
&\leq& 6\E_S[c(T_S^*)] + E_R[c(MST(R))] \leq 8\E_S[c(T_S^*)].
\end{eqnarray*}  
Let $T_{PD}(R^*)$ be the Steiner tree over $R^*\cup \{t\}$ computed by the primal-dual algorithm \cite{GW95}. Then, 
$$EST(R^*,V\setminus R^*) = 2\bar{c}_{ST}(R^*) + \sum_{v\in V} p_v D_v(R^*)\geq c(T_{PD}(R^*)) + 
\E_S\left[\sum_{v\in S} D_v(R^*)\right].$$ 
By combining inequalities \eqref{boundStrategies} and \eqref{T_RA} (after replacing $R$ by $R^*$ and $T_{\alpha}(R^*)$ by $T_{PD}(R^*)$) with all the above, we have that  
\begin{eqnarray*}
\E_S[c(\P_{R^*}(S))] \leq 2 \left(c(T_{PD}(R^*)) + \E_S\left[\sum_{v\in S} D_v(R^*)\right]\right) \leq 2EST(R^*,V\setminus R^*) \leq 16\E_S[c(T_S^*)].
\end{eqnarray*}
 
\end{proof}


\newpage
\bibliographystyle{plain} 
\bibliography{cost-sharing}

\end{document}